\documentclass[11pt,pdftex,a4paper]{article}

\usepackage{graphicx} % Required for inserting images

\usepackage{amsmath}
\usepackage{amsthm}
\usepackage{mathtools}

\usepackage{hyperref}

\usepackage{fullpage}

\usepackage{listings}
\usepackage{xcolor}

\usepackage[boxed,linesnumbered]{algorithm2e} % I added this :)
\usepackage{tikz} % this too :)

\newtheorem{theorem}{Theorem}
\newtheorem{lemma}[theorem]{Lemma}
\newtheorem{definition}{Definition}

\newcommand{\myparagraph}[1]{\vspace{2mm}\noindent{\textbf{#1.}}}

\def\REVIEW{0}  %% Enables review commands such as \xxrev
\if \REVIEW 1
    \usepackage[enable]{xnotes}
\else
    \usepackage[apply,nonotes]{xnotes}
\fi

\AddXNotesUser{ns}{Nathan}{blue}
\AddXNotesUser{pk}{Petr}{red}

\title{Resolving Conflicts with Grace:\\ Dynamically Concurrent Universality}

\author{Petr Kuznetsov$^1$ \and Nathan Josia Schrodt$^2$}
\date{%
    $^1$ LTCI, Télécom Paris, Institut Polytechnique de Paris, France\\%
    $^2$Télécom Paris, Institut Polytechnique de Paris, France \& Technical University of Darmstadt, Germany% 
}

\begin{document}
\maketitle

\begin{abstract}
    Synchronization is the major obstacle to scalability in distributed computing.
    Concurrent operations on the shared data engage in synchronization when they encounter a \emph{conflict}, i.e., 
    their effects depend on the order in which they are applied.
    Ideally, one would like to detect conflicts in a \emph{dynamic} manner, i.e., adjusting to the current system state. 
    Indeed, it is very common that two concurrent operations conflict only in some rarely occurring states.
    In this paper,
    we define the notion of \emph{dynamic concurrency}:
    an operation employs strong synchronization primitives only if it \emph{has} to arbitrate with concurrent operations, given the current system state.
    We then present a dynamically concurrent universal construction.
\end{abstract}

%\begin{center}
%    {\bf Eligible for the best student paper award}
%\end{center}

%\clearpage
%\tableofcontents
%\clearpage
\section{Introduction}
When we design a distributed system that enables concurrent access to the shared data, we often have to face our main nemesis -- synchronization.
It is common that concurrent operations affect one another and, thus, cannot be executed in parallel;
we say that the operations encounter a \emph{conflict}.
Potential conflicts are typically anticipated and resolved using \emph{consensus}~\cite{flp}.
The processes agree on the order in which the %conflicting
operations are applied, as in the celebrated universal construction by Herlihy~\cite{Her91}. 

As consensus is perceived an expensive tool, a lot of efforts were invested in \emph{optimistic} algorithms that avoid consensus-based synchronization in \emph{good scenarios}, in particular when no contention is encountered, i.e., when an operation runs in isolation~\cite{paxos,HLM03-of,solofast}.
%, both in distributed~\cite{} and concurrent~\cite{} systems.  
%
A step further in this direction is to avoid using consensus \emph{in the absence of conflicts}, even when operations contend.
Indeed, one can come up with a universal construction in which concurrent operations that do not conflict are executed in parallel without engaging in consensus to decide on their execution order~\cite{lamport2010generalized,bazzi2022clairvoyant}.  

However, conflicts are treated here in a conservative way: two operations are considered conflicting if they do not commute \emph{in some state} (but might commute in another).
%
%Consensus (or other strong synchronization primitives~\cite{solofast}) is then used to order such operations, even if the current system state allows them to commute.    
This can result in very inefficient implementations.
%The presence of a ``real'' conflict is determined on the current \emph{state} of the shared data.   
%
%
Indeed, many data structures exhibit conflicts \emph{rarely}, only in certain states.   
Consider, for example an \emph{asset-transfer} system~\cite{at2-cons} that allows a set of users to execute \emph{transactions} sending assets across their accounts. 
When two transactions $\textit{transfer}(\$100)$ and $\textit{transfer}(\$50)$ are concurrently applied on the same account holding $\$100$, their effect depends on the \emph{order} in which they are executed:  the first one will be accepted, and the second one -- rejected.
On the other hand, if the account holds at least $\$150$, the two transactions can be applied in any order with the same effect.

The notion of a dynamic conflict in asset-transfer systems is defined in a straightforward manner, and one can come up with a system that uses consensus only \emph{when necessary}, i.e., concurrent transactions  cannot \emph{all} be served~\cite{cryptoconcurrency}.
To the best of our knowledge, however, no attempt has been taken to define dynamic conflicts in a \emph{universal} way, that would apply to \emph{any} sequential object.
%
%In this paper, we focus on \emph{dynamic}, state-dependent conflicts.
%
We address this issue in the concurrent model, where processes communicate by applying operations on shared \emph{base} objects, heading for a \emph{universal construction}, a \emph{wait-free linearizable}~\cite{linearizability,Her91} implementation of any object, given its sequential specification. 
We distinguish between basic read-write operations (offering \emph{consensus power} $1$~\cite{Her91,flp,LA87}) and \emph{strong} (read-modify-write) synchronization primitives that enable nontrivial consensus solutions, such as CAS or LL/SC. 
Our goal is a \emph{dynamically concurrent} universal construction that avoids using strong synchronization if concurrent operations commute \emph{in the current system state}. 
The goal, quite easily grasped on an intuitive level, is however evasive when it comes to formal treatment.
How is the ``current state'' defined in an execution of a concurrent algorithm?

Indeed, consider an operation in a concurrent history. 
The operation's interval may be concurrent with different operations at different moments of its execution.  
How do we define the state of the implemented object and decide whether the operation commutes with concurrent operations in that state? 
Imagine an implementation in which every operation $op$ can be thought of taking effect at a well-defined moment between its invocation and its response, a \emph{linearization point}~\cite{linearizability}.
Then, each moment during the execution of $op$ determines a system state, consisting of the effects of operations that have already taken place and the set of operations that have been invoked but not yet linearized. 
Even if the actual implementation does not allow us to determine an irreversible linearization point for each operation~\cite{strong-lin}, we can derive a superset of all possible system states that \emph{might have occurred}. 
Intuitively, an implementation is said to be \emph{dynamically concurrent} if it uses strong synchronization for an operation only when, in some possible system state during its execution, it does not commute with some concurrent operations (see Figure~\ref{fig:history} for an illustration).

The canonical approach to building a universal construction is to order operations using consensus and applying the sequential specification of the object to the ordering. 
A natural generalization of this approach, aimed at allowing concurrency, is to replace the linear model with a directed acyclic \emph{dependency graph} of operations, in which all topological orderings yield equivalent linearizations. 
This way commuting operations are executed in parallel, which might considerably improve performance in workloads where dynamic conflicts are rare, often argued to be the most common case.

Our dynamically concurrent universal construction builds on this idea using a shared dependency graph together with a set of active operations. 
The current state is obtained by topologically ordering the current version of the graph, while the set of active operations is determined by identifying those that have been invoked but not yet added to the graph. 
If commutativity with the concurrent operations is observed, the current operation can be appended to the graph directly; otherwise, the algorithm resorts to ordering the operations via consensus. 
A particular difficulty here is to ensure that conflicting operations committed via different paths (direct commitment and consensus-based) cannot bypass each other, preserving correctness.
We achieve this by a variant of commit-adopt~\cite{Gaf98} that guarantees that a directly committed operation is always placed in the causal past of every concurrent operation that encounters a conflict with it.

To sum up, in this paper we:

\begin{enumerate}
    \item[(1)] Define the notion of a \emph{dynamically concurrent} object, the one that resorts to consensus only when a dynamic conflict arises.
    \item[(2)] Describe a dynamically concurrent universal construction.
\end{enumerate}

The paper is organized as follows. 
In Section~\ref{sec:related}, we overview the related work and 
in Section~\ref{sec:model}, we recall our model assumptions. 
In Section~\ref{sec:dyncon}, we formally define the notion of dynamic concurrency. 
In Section~\ref{sec:protocol}, we present our protocol and analyze its correctness. 
Section~\ref{sec:conc} concludes the paper with the discussion of ramifications of this work and open questions. 
Certain proofs are delegated to the appendix.

\section{Related Work}
\label{sec:related}

%Distributed systems are subject to the fundamental trade-off between consistency and availability~\cite{bre12cap,GL2}. 
% 
Most concurrent data structures exhibiting strong consistency and availability guarantees, such as linearizability~\cite{linearizability,AW04} and wait-freedom~\cite{Her91}, require using of consensus or strong synchronization primitives.   

As consensus often becomes an implementation's bottleneck, a lot of effort has been invested in \emph{optimistic} solutions that avoid consensus in ``good runs'', in particular, when no contention is observed~\cite{hq,zyzzyva,hotstuff}.
In concurrent data structures, this approach instantiated in \emph{solo-fast} implementations~\cite{solofast}, where strong synchronization is only used when \emph{step contention} is observed.

%Optimistic approach (go fast if no contention) - Q/U replication, HQ replication, Zyzzyva, HotStuff, Leaderless consensus (egalitarian paxos, sutra), optimistic generic broadcast (zielinski)

Schneider~\cite{Sch90} introduced the notion of a \emph{state-machine replication}, a concurrent implementation of an object given its sequential specification.
Lamport~\cite{paxos} discovered the first fault-tolerant state-machine replication protocol and then generalized it to account for parallel execution of \emph{non-conflicting operations}~\cite{lamport2010generalized}.
In the same vein, Zielinski proposed an elegant and fast \emph{generic broadcast} protocol that ensures that conflicting broadcast messages are delivered in the same order (non-conflicting ones do not have to be ordered).
Both proposals assume a \emph{static} definition of conflicts, where certain pairs of commands are conservatively considered as conflicting if they do not commute in \emph{some state}.
The static understanding of conflicts has been also considered in the Byzantine fault-tolerant setting: in Byzantine generalized Paxos~\cite{byzgenpaxos}, Byblos~\cite{bazzi2022clairvoyant} and RedBlue~\cite{li2012redblue}.

Aspnes and Herlihy \cite{Aspnes90} presented a wait-free construction for objects where every operations either commutes with or overwrites any other operation in any history. Like our approach, their method uses a precedence graph to capture dependencies between operations. Requiring that every operation either commutes or overwrites any other is a major restriction, since it limits the construction to objects with consensus number one. In contrast, our construction supports arbitrary objects and detects conflicts dynamically.

The notion of dynamic (state-dependent) conflicts has also been considered by Clements et al.~\cite{ClementsKZMK13} in their work on \emph{commutativity rule}, an efficient generalization of disjoint-access-parallelism by Israeli and Rappoport~\cite{dap} to semantics-aware systems.     
Kulkarni et al.~\cite{kulkarni2009defining,KulkarniNPSP11} discuss algorithms that achieve serializability of transactions and adjust to \emph{low-level} dynamic conflicts incurred by underlying \emph{base-object} operations: transactions that do not encounter low-level data conflict proceed in parallel, without unnecessary lock-based synchronization. 
In contrast, we focus here on \emph{high-level semantics-aware} dynamic conflicts and \emph{wait-free} implementations that do not engage in consensus in conflict-free scenarios.  

%We use a state-depended notion of commutativity, as seen for example in~\cite{ClementsKZMK13},
%
%Unlike definitions that apply to entire sets of operations or consider commutativity pairwise, we define whether an individual operation commutes with a set of operations.

To the best of our knowledge, the only attempt to account for dynamic conflicts in this context was taken in CryptoConcurrency~\cite{cryptoconcurrency}, an asset transfer protocol that processes concurrent transactions on the same account in parallel, without resorting to consensus to order them, as long as they are not exhausting the current account balance. 
In this paper, we generalize this idea to universal construction that allows concurrent operations to proceed in parallel, as long as they are not in conflict in the current object state.

\section{Model}
\label{sec:model}
We consider a system of $n$ asynchronous processes. Each process can be arbitrarily fast or slow relative to the other processes and can \emph{crash fail}, i.e., stop taking steps at any moment. The processes communicate through fundamental shared base objects like atomic read-write registers. Those base objects can be used to implement more complex, shared, \emph{sequential objects}.

\myparagraph{Objects}
A \emph{sequential object $X$} is defined by a tuple $(Q, q_0, O, R, \sigma)$ where $Q$ is a set of states, $q_0 \in Q$ is the initial state, $O$ is a set of operations, $R$ is a set of outputs, and $\sigma: O \times Q \to R \times Q$ is a sequential specification with $\sigma((o, q)) = (r, q')$ if and only if operation $o$ applied to $X$ in state $q$ returns $r$ and changes its state to $q'$.
Since we consider only sequential objects, we will henceforth refer to them simply as \emph{objects}.

\myparagraph{Steps and Algorithms}{
The behavior of a process when implementing a \emph{high-level object} (as opposed to a base object) is described by a deterministic automaton that specifies how each process responds to invocations of that object. Upon receiving an invocation, the process advances by taking \emph{steps} prescribed by its automaton. We refer to the descriptions of these automatons simply as \emph{implementations} or \emph{algorithms}.
Each step consists of some terminating local computation together with a single operation on a shared base object. Infinite loops of purely local computation without any interaction with a base object are disallowed.
In addition, an algorithm may invoke operations of other higher-level shared objects. Such invocations act as subautomata: control is delegated to the invoked object until its operation finishes, after which the original automaton continues. Conceptually, this recursive structure can be flattened, leaving only local computations and operations on base objects.
After each step, a process updates its local state according to its transition rules and may return a response to the currently active operation of the high-level object being implemented. Crucially, a process is not allowed to start executing a new operation on this implemented object until the previous operation has terminated.
}

\myparagraph{Runs and Operations}{
A \emph{run} is a (possibly infinite) sequence of events obtained by interleaving the steps of all processes, beginning from the initial global state.
A run especially includes every taken step as well as invocations and response events of operations of every (high-level or base) object involved.
For simplicity, we assume that each operation is unique, i.e., an operation is invoked at most once.
This uniqueness can be ensured by assigning a caller-specific identifier and a sequence number to every operation.
We denote the process executing an operation $op$ as $op.p$, and refer to the invocation and response (event) of $op$ as $op.inv$ and $op.res$, respectively.
Invocations and responses can be uniquely associated with operations.
}

\myparagraph{Histories and Linearizability}
A \emph{history} is a sequence $H$ of invocations and responses without duplicates, where every response $op.res$ is preceded by its unique matching invocation $op.inv$, and each invocation $op.inv$ corresponds to at most one response. We define $op \in H$ if and only if $op.inv \in H$.
A run \emph{induces} a history for a given object by listing every invocation and response of operations on that object that occur in the run, in the order in which they appear.
A history is \emph{sequential} if every invocation can only be followed by a corresponding response. Thus, a sequential history $S$ is a list of operations and vice versa. 
Whenever $S$ is finite, the last operation in $S$ might not have a response. 

Operation $op$ \emph{precedes $op'$ in $H$} if $op.res$ precedes $op'.inv$ in $H$. We write $op <_H op'$. If neither $op <_H op'$ nor $op' <_H op$ but $op, op' \in H$, $op$ and $op'$ are \emph{concurrent in $H$}.
The relation $<_H$ is irreflexive, asymmetric and transitive and therefore forms a strict partial order. For a sequential history, $<_H$ is a total order. We refer to $<_H$ as the \emph{precedence relation of $H$}.
A sequential history $S$ is \emph{legal} if it satisfies the sequential specifications of the objects, i.e., for each object, starting from its initial state and sequentially applying the operations of $S$ that involve this object yields exactly the responses recorded in $S$.
An operation $op$ is \emph{complete in $H$} if $H$ contains both the invocation and the response of $op$. A \emph{completion of $H$} is a history $H'$ that only includes all complete operations of $H$ and a subset of incomplete operations of $H$ with corresponding responses.
Furthermore, $H'$ has to preserve the precedence relation of $H$.
Let us denote the history that consists of all invocations and responses of a process $p$ in a history $H$ as $H \mid p$.

\begin{definition}
    \label{def:history.equivalent}
    Histories $H$ and $H'$ are \textbf{equivalent} if we have $H \mid p = H' \mid p$ for all processes $p$.
\end{definition}

\begin{definition}
    \label{def:history.linearization}
    Let $H$ be a history. A sequential history $S$ is a \textbf{linearization of $H$} if
    \begin{itemize}
        \item $S$ is equivalent to a completion of $H$,
        \item $S$ preserves the precedence relation of $H$ and
        \item $S$ is legal.
    \end{itemize}
    If there exist such $S$ for a history $H$, we call $H$ \textbf{linearizable}.
\end{definition}

An implementation (or algorithm) $I$ is \emph{linearizable} if, for every run $R$ of $I$, the history induced by $R$ is linearizable.
Linearizability is a widely accepted correctness condition for concurrent objects. It requires that every operation on the object appears to take effect at a single point in time between its invocation and response, as though the object was accessed atomically.

\myparagraph{Wait-freedom}
The gold standard for progress guarantees is \emph{wait-freedom}.
An implementation of an object is \emph{wait-free} if, whenever a process that does not crash fail invokes an operation, that process is guaranteed to complete the operation in a finite number of its own steps, regardless of the actions or failures of other processes.

\myparagraph{Snapshots and Consensus}
A fundamental object with a wait-free linearizable implementation using only atomic read-write registers is the \emph{snapshot} object \cite{snapshot}.
Each process has a dedicated component it can update with a $\textit{write}$ operation, and all processes can invoke $\textit{snapshot}()$ to obtain a consistent view of all components.
Another key object is \emph{consensus}, exporting a single operation $\textit{propose}(v)$.
Each process can propose values $v$ from a set $V$, and all invocations must return the same value (\emph{agreement}), which must be one of the proposed values (\emph{validity}).
We say that a consensus object \emph{stores} a value $v$, once an invocation of $\textit{propose}$ returned that value.
Agreement ensures that no other $\textit{propose}$ operation on the same object can return a different value.
Unlike the snapshot object, consensus cannot be implemented in a wait-free linearizable manner using only read-write registers \cite{flp}.
Any base object that, together with read-write registers, enables solving wait-free consensus, such as CAS or LL/SC, is called a \emph{strong synchronization primitive}.

\section{Dynamic Concurrency}
\label{sec:dyncon}

We now introduce the notion of a dynamically concurrent implementation.
Before formalizing the notions of the system state and dynamic concurrency in Definitions~\ref{def:system-state} and~\ref{def:dynamic-concurrency}, we define the equivalence and commutativity relations and then provide some intuition on dynamic concurrency.

\subsection{Equivalences and Commutativity}
In this section, we define equivalences between sequences of operations, as well as the notion of commutativity,
crucial in formalizing the concept of dynamic concurrency in the next section and, accordingly, in determining whether certain operations require ordering in a dynamically concurrent universal construction.
Let $s$ and $s'$ be disjoint sequences of operations. We denote $s \cdot s'$ as the concatenation of $s$ and $s'$. A single operation is interpreted as a sequence of length one.
\begin{definition}
    \label{def:equivalence-relation-of-orderings}
    Let $s$ and $s'$ be orderings of a set of operations $S$. The orderings $s$ and $s'$ are \textbf{equivalent} if according to the sequential specification the application of $s$ to the initial state results in the same state as the application of $s'$ and every operation in $S$ returns the same output in both orderings. We write $s \simeq s'$.
    
    If $s \simeq s'$, and $s$ and $s'$ share a common prefix $p$ such that $s = p \cdot r$ and $s' = p \cdot r'$ for some $r$ and $r'$, we say that $r$ and $r'$ are \textbf{equivalent in $p$}.
\end{definition}

We immediately see that the relation $\simeq$ of sequences of operations, defined in Definition~\ref{def:equivalence-relation-of-orderings}, is a well-defined equivalence relation for a given set of operations, i.e., it is reflexive, symmetric and transitive.

Two states are said to be \emph{distinguishable} if there exists a sequence of operations whose execution from these states yields different responses.
If certain states cannot be distinguished in this way, the object's states can be replaced by equivalence classes, where two states are equivalent if they are not distinguishable.
Working with this quotient state space can increase the number of equivalent orderings and, potentially, the scenarios in which operations commute according to Definition~\ref{def:commutativity}.

\begin{definition}
    \label{def:commutativity}
    Let $op$ be an operation and $l$ and $s$ sequences of operations with $op \notin l$, $op \notin s$ and $l \cap s = \emptyset$.
    Operation $op$ \textbf{commutes in $l$ with $s$} if $l \cdot s \cdot op \simeq l \cdot op \cdot s$.

    Let $S$ be a set of operations with $op \notin S$ and $l\cap S = \emptyset$.
    Operation $op$ \textbf{commutes in $l$ with $S$} if $op$ commutes in $l$ with every ordering of $S$.
\end{definition}
\begin{lemma}
    \label{lemma:equivalence-preservation}
    Let $l$ be a sequence of operations, $s = s_1 \cdot s_2$ and $s' = s_1' \cdot s_2'$ orderings of the same set of operations $S$, and $op$ an operation with $op \notin l$, $op \notin S$ and $l \cap S = \emptyset$.
    
    If $l \cdot s \simeq l \cdot s'$ and $op$ commutes with $s_1$, $s_1'$, $s$ and $s'$ in $l$, then it follows that $l \cdot s_1 \cdot op \cdot s_2 \simeq l \cdot  s_1' \cdot op \cdot s_2'$.
\end{lemma}
\begin{proof}
    The equivalence follows by transitivity of the equivalence relation.
    
    $l \cdot s_1 \cdot op \cdot s_2 \simeq l \cdot op \cdot s_1 \cdot s_2 \simeq l \cdot s_1 \cdot s_2 \cdot op \simeq l \cdot s_1' \cdot s_2' \cdot op \simeq l \cdot op \cdot s_1' \cdot s_2' \simeq l \cdot s_1' \cdot op \cdot s_2'$
\end{proof}

A key result of Lemma~\ref{lemma:equivalence-preservation} with $s=s'$, used throughout this paper, is that if an operation commutes with every subset of a set, every placement of this operation in this set is equivalent.

\subsection{Definitions}
The idea behind a linearization, or linearizability in general, is that each operation can be viewed as taking effect atomically at a single moment between its invocation and its response. In many cases, this perspective is straightforward, as one can identify a step where an operation takes effect. In other cases, however, this perspective serves only as an abstraction, since no single concrete point in time can be pinpointed.

Our definition uses this abstraction to define possible configurations during the execution of an operation. Suppose that a linearizable history $H$ is assigned a \emph{linearization point} to each of its operations, i.e., the moment between its invocation and its response where the operation an be thought of taking effect. Then, the notion of a system state during an operation $op$ follows naturally. Take a prefix of that enriched history that includes the invocation of $op$ but not its linearization point. The system state consists of two elements.
\begin{enumerate}
    \item The sequence of operations according to their linearization points.\footnote{Note that for a given sequential specification a sequence of operations defines the state of the object.}
    \item The set of all operations that have been invoked but not linearized.
\end{enumerate}

\begin{definition}
    \label{def:system-state}
    Let $H$ be a history and $op$ an operation.
    A \textbf{system state during $op$ in $H$} is a pair $(l,O)$, where $l$ is a sequence and $O$ a set of operations, such that:
    \begin{enumerate}
        \item There exists a prefix $H'$ of $H$ with $op.inv \in H'$ and $op.res \notin H'$.
        %\item There exists a linearization $L$ of $H$ such that $l$ is a prefix of $L$ and $l$ contains all operations $op'$ with $op'.res \in H'$, but not $op$ itself.
        %\item For all $op' \in l$, it holds that $op'.inv \in H'$.
        \item There exists a linearization $L$ of $H$ such that $l$ is a prefix of $L$ and a linearization of $H'$ that does not include $op$.
        \item The set $O$ consists precisely of those operations $op'$ for which $op'.inv \in H'$, $op' \notin l$, and $op' \neq op$.
    \end{enumerate}
\end{definition}

Note that the linearization $l$ in Definition~\ref{def:system-state} may also contain operations whose responses are not included in $H'$, provided they precede $op$ in $L$.

With the definition of system states in place, the notion of dynamic concurrency can be introduced. 
Intuitively, a dynamically concurrent implementation of an object uses a strong synchronization primitive while executing an operation $op$, only when, during its execution, it could have been in a system state where $op$ does not commute with every subset of the concurrent operations in that current state.

\begin{definition}
    \label{def:dynamic-concurrency}
    An implementation $I$ of an object is \textbf{dynamically concurrent} if for every history $H$ induced by a run of $I$ and every operation $op$ in $H$, the following holds:

    Process $op.p$ uses a strong synchronization primitive while executing $op$.
    
    $$\Longrightarrow$$
    
    There exists a system state $(l, O)$ during $op$ in $H$ such that there exists a subset $O' \subseteq O$ with which $op$ does not commute in $l$.
\end{definition}

While the above definitions formalize the concepts of system state and dynamic concurrency, their intuition may not be immediately apparent. The following example demonstrates how these notions manifest in a concrete execution.

\subsection{Example}
To illustrate the concepts of system state and dynamic concurrency, consider the history $H$ depicted in Figure~\ref{fig:history}. 
The history is induced by a run of $4$ processes executing operations on a shared \emph{list} object.
This object exports the methods $\textit{append}(v)$, $\textit{readLast}()$, $\textit{readAll}()$ and $\textit{swap}(i,j)$ for $v\in\{a,b,c,d\}$ and $0 \leq i \leq j$. The method $\textit{append}(v)$ appends $v$ at the end of the list and returns $ok$; $\textit{readLast}()$ returns the value at the last index of the list, $\bot$ for an empty list; $\textit{readAll}()$ returns the whole list; $\textit{swap}(i,j)$ swaps the values at indices $i$ and $j$ and returns $ok$ if the list has at least $j-1$ elements, otherwise it has no effect and returns $\bot$.

%%%%%%%%%%%%%%%%%%%%%%%%%%%%%%%%%%
%%%% beautiful drawing below %%%%%
%%%%%%%%%%%%%%%%%%%%%%%%%%%%%%%%%%

\begin{figure}[htbp]
    \centering
    \resizebox{\dimexpr0.9\textwidth}{!}{
        \begin{tikzpicture}[every node/.style={font=\scriptsize}, yscale=1.35, xscale=0.5]
        
        % process labels
        \foreach \i in {0,...,3} {\node[left, font=\small] at (0,-\i) {$p_\i$};}
        
        % p_0
        \draw (0,0) -- (13,0);
        
        \node (p0a) at (13,0) {[};
        \node[above=3pt, anchor=south west, xshift=-5pt] at (p0a) {$\textit{append}(c)$};
        \draw[->, dotted] (13,0) -- (17,0);
        
        % p_1
        \draw (0,-1) -- (4,-1);
        
        \node (p1a) at (4,-1) {[};
        \node[above=3pt, anchor=south west, xshift=-5pt] at (p1a) {$\textit{append}(b)$};
        \draw[dotted] (4,-1) -- (6,-1);
        \node (p2b) at (6,-1) {]};
        \node[below=3pt] at (p2b) {$[ok]$};
        
        \draw (6,-1) -- (12,-1);
        
        \node (p1c) at (12,-1) {[};
        \node[above=3pt, anchor=south west, xshift=-5pt] at (p1c) {$\textit{append}(d)$};
        \draw[dotted] (12,-1) -- (15,-1);
        \node (p2d) at (15,-1) {]};
        \node[below=3pt] at (p2d) {$[ok]$};
        
        \draw[->] (15,-1) -- (17,-1);
        
        % p_2
        \draw (0,-2) -- (1,-2);

        \node (p2a) at (1,-2) {[};
        \node[above=3pt, anchor=south west, xshift=-5pt] at (p2a) {$\textit{append}(a)$};
        \draw[dotted] (1,-2) -- (3,-2);
        \node (p2b) at (3,-2) {]};
        \node[below=3pt] at (p2b) {$[ok]$};
        
        \draw (3,-2) -- (9,-2);
        
        \node (p2c) at (9,-2) {[};
        \node[above=3pt, anchor=south west, xshift=-5pt] at (p2c) {$\textit{readLast}()$};
        \draw[dotted] (9,-2) -- (11,-2);
        \node (p2d) at (11,-2) {]};
        \node[below=3pt] at (p2d) {$[a]$};
        
        \draw (11,-2) -- (14,-2);
        
        \node (p2e) at (14,-2) {[};
        \node[above=3pt, anchor=south west, xshift=-5pt] at (p2e) {$\textit{readAll}()$};
        \draw[dotted] (14,-2) -- (16,-2);
        \node (p2f) at (16,-2) {]};
        \node[below=3pt] at (p2f) {$[[a,b,a]]$};
        
        \draw[->] (16,-2) -- (17,-2);
        
        % p_3
        \draw (0,-3) -- (5,-3);
        
        \node (p3a) at (5,-3) {[};
        \node[above=3pt, anchor=south west, xshift=-5pt] at (p3a) {$\textit{append}(a)$};
        \draw[dotted] (5,-3) -- (7,-3);
        \node (p3b) at (7,-3) {]};
        \node[below=3pt] at (p3b) {$[ok]$};
        
        \draw (7,-3) -- (10,-3);
        
        \node (p3c) at (10,-3) {[};
        \node[above=3pt, anchor=south west, xshift=-5pt] at (p3c) {$\textit{swap}(0,2)$};
        \draw[->, dotted] (10,-3) -- (17,-3);

        % linearization points
        % p1: append(b)
        \fill (5,-1) circle (3pt);
        
        % p2: append(a)
        \fill (2,-2) circle (3pt);
        
        % p3: append(a)
        \fill (6,-3) circle (3pt);

        \end{tikzpicture}
    }
    \caption{An illustration of a history $H$ with $4$ processes using a shared list object. The response of $p_2.\textit{readAll}()$ requires every linearization to order $p_1.\textit{append}(b)$ before $p_3.\textit{append}(a)$. Note that $p_3.\textit{swap}(0,2)$ commutes with $p_2.\textit{readLast}()$ in $\big(p_2.\textit{append}(a),p_1.\textit{append}(b),p_3.\textit{append}(a)\big)$. Moreover, for every system state during $p_3.\textit{swap}(0,2)$ in $H$, $p_3.\textit{swap}(0,2)$ commutes with every subset of concurrent operations in the respective state of the object. Dynamic concurrency forbids $p_3$ to use a strong synchronization primitive during the execution of $p_3.\textit{swap}(0,2)$.}
    \label{fig:history}
\end{figure}
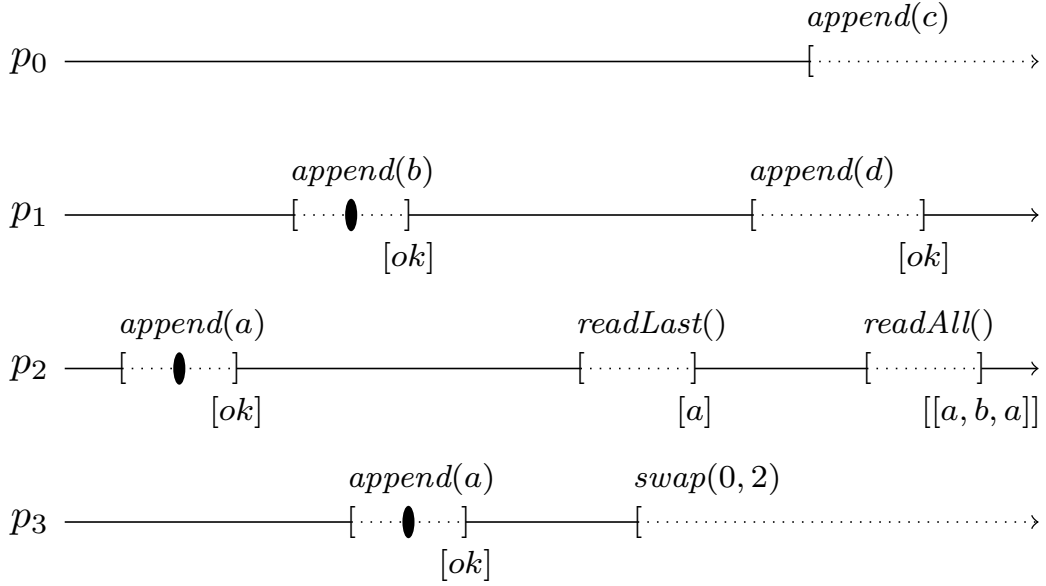

%%%%%%%%%%%%%%%%%%%%%%%%%%%%%%%%%%
%%%% beautiful drawing above %%%%%
%%%%%%%%%%%%%%%%%%%%%%%%%%%%%%%%%%

Let us now examine the system states during $op := p_3.\textit{swap}(0,2)$ in $H$.
To determine these states, we consider all possible linearizations $L$ of $H$.
The response of $p_2.\textit{readAll}()$ requires that in every linearization, $p_1.\textit{append}(b)$ is ordered before $p_3.\textit{append}(a)$.
Hence, $\big(p_2.\textit{append}(a)$, $p_1.\textit{append}(b)$, $p_3.\textit{append}(a)\big)$ forms a prefix of every linearization.

This constraint stems from the choice made in Definition~\ref{def:system-state} to require $l$ to be a prefix of a linearization of the entire history $H$, rather than simply a linearization of a prefix $H'$. The motivation for this choice is to obtain a more refined representation of the implementation's internal state. In particular, the definition assumes that after certain operations have terminated (such as $p_1.\textit{append}(b)$ and $p_3.\textit{append}(a)$ in our example), the internal state of the implementation has either already determined their relative order, even though this ordering is not yet externally observable, or the implementation will never establish such an order at all (in the case of equivalent orderings).

Furthermore, in every linearization $p_2.\textit{readLast}()$ precedes $p_2.\textit{readAll}()$, which in turn precedes $p_0.\textit{append}(c)$ (if present) and $p_1.\textit{append}(d)$.
Excluding $op$ for the moment, the possible linearizations are:
\begin{enumerate}
    \item $\textit{pre} \cdot \big(p_0.\textit{append}(c),p_1.\textit{append}(d)\big)$,
    \item $\textit{pre} \cdot \big(p_1.\textit{append}(d),p_0.\textit{append}(c)\big)$,
    \item $\textit{pre} \cdot \big(p_1.\textit{append}(d)\big)$,
\end{enumerate}
where $\textit{pre}:=\big(p_2.\textit{append}(a),p_1.\textit{append}(b),p_3.\textit{append}(a),p_2.\textit{readLast}(),p_2.\textit{readAll}()\big)$.
All $14$ linearizations that include $op$ can be obtained by inserting it at any position after the first three operations.
However, since we are only interested in those prefixes of linearizations that do not include $op$, we can simply omit it.
According to Definition~\ref{def:system-state}, we obtain the following system states during $op$ in $H$.
\begin{enumerate}
    \item $l_1=\big(p_2.\textit{append}(a),p_1.\textit{append}(b),p_3.\textit{append}(a)\big)$ with $O_1=\big\{p_2.\textit{readLast}()\big\}$,
    \item $l_2=\big(p_2.\textit{append}(a),p_1.\textit{append}(b),p_3.\textit{append}(a),p_2.\textit{readLast}()\big)$ with $O_2=\big\{\big\}$,
    \item $l_3=\big(p_2.\textit{append}(a),p_1.\textit{append}(b),p_3.\textit{append}(a),p_2.\textit{readLast}()\big)$ with $O_3=\big\{p_1.\textit{append}(d)\big\}$,
    \item $l_4=\big(p_2.\textit{append}(a),p_1.\textit{append}(b),p_3.\textit{append}(a),p_2.\textit{readLast}()\big)$\\
    with $O_4=\big\{p_1.\textit{append}(d),p_0.\textit{append}(c)\big\}$,
    \item $l_5 =\big(p_2.\textit{append}(a),p_1.\textit{append}(b),p_3.\textit{append}(a),p_2.\textit{readLast}()\big)$\\
    with $O_5 =\big\{p_1.\textit{append}(d),p_0.\textit{append}(c),p_2.\textit{readAll}()\big\}$,
    \item $l_6 =\textit{pre}$ with $O_6 =\big\{p_1.\textit{append}(d),p_0.\textit{append}(c)\big\}$,
    \item $l_7 =\textit{pre} \cdot \big(p_0.\textit{append}(c)\big)$ with $O_7 =\big\{p_1.\textit{append}(d)\big\}$,
    \item $l_8 =\textit{pre} \cdot \big(p_1.\textit{append}(d)\big)$ with $O_8 =\big\{p_0.\textit{append}(c)\big\}$,
    \item $l_9 =\textit{pre} \cdot \big(p_0.\textit{append}(c),p_1.\textit{append}(d)\big)$ with $O_9 =\big\{\big\}$,
    \item $l_{10} =\textit{pre} \cdot \big(p_1.\textit{append}(d),p_0.\textit{append}(c)\big)$ with $O_{10} =\big\{\big\}$.
\end{enumerate}

According to Definition~\ref{def:dynamic-concurrency}, dynamic concurrency forbids $p_3$ to use strong synchronization primitives during the execution of $op$, since in every system state $(l_i, O_i)$ $op$ commutes in $l_i$ with every subset of $O_i$.\footnote{The condition for all $i \in [1,10]$ is equivalent to requiring it only for $i \in \{1,5\}$.} Note that this holds even though $p_0.\textit{append}(c)$ and $p_1.\textit{append}(d)$ do not commute in any state.

Similarly, $p_2$ cannot use strong synchronization primitives during the execution of either $p_2.\textit{append}(a)$ or $p_2.\textit{readLast}()$.

As discussed above, the response of $p_2.\textit{readAll}()$ enforces that $p_1.\textit{append}(b)$ must be ordered before $p_3.\textit{append}(a)$. Suppose, however, that $p_2.\textit{readAll}()$ were absent from the history. In that case, $p_3.\textit{append}(a)$ could instead be ordered before $p_1.\textit{append}(b)$, since the response of $p_2.\textit{readLast}()$ would remain correct provided that $p_3.\textit{swap}(0,2)$ is ordered before $p_2.\textit{readLast}()$. In this scenario, $p_3.\textit{swap}(0,2)$ no longer commutes with $p_2.\textit{readLast}()$, and consequently $p_3$ would be permitted to use strong synchronization primitives while executing $p_3.\textit{swap}(0,2)$.

\section{Dynamically Concurrent Universal Construction}
\label{sec:protocol}
We now present our (wait-free and linearizable) dynamically concurrent universal construction. 
The universal construction takes the specification of a sequential object as a parameter and exports one method $\textit{publish}(op)$ that simulates the execution of an operation $op$ and returns a response of $op$.

A typical universal construction invokes multiple instances of consensus to agree on the order in which invoked operations are linearized. 
The response of each operation is then computed from this order using the object's sequential specification.

To exploit commutativity, we replace the linear ordering on operations with a shared \emph{dependency graph}. 
Intuitively, in the graph, an operation has incoming edges from every operation that should precede it.
Unrelated operations in the graph are then concurrent and commuting.
A linearization is obtained by topologically ordering this dependency graph.

For this algorithm to be correct, we essentially have to ensure the following conditions.
First of all, after the completion of $\textit{publish}(op)$ the operation $op$ has to be added to the dependency graph.
Secondly, this dependency graph has to be an acyclic graph to ensure that topological orderings exist.
Another requirement is that the topological orderings, and therefore the dependency graph, should respect the \emph{real-time ordering} of different invocations of $\textit{publish}$, i.e., the precedence relation of the induced history.
If $\textit{publish}(op)$ terminates before $\textit{publish}(op')$ is invoked, $op$ should precede $op'$ in every topological ordering, i.e., there must be an edge from $op$ to $op'$ in the dependency graph.
Furthermore, all possible topological ordering of this dependency graph should be \emph{equivalent}, as stated in Definition~\ref{def:equivalence-relation-of-orderings}. 
Intuitively, unrelated graph vertices should commute and, thus, no process can distinguish their different orderings.
Finally, we have to guarantee that this dependency graph is \emph{stable} in the sense that no response of an operation, according to any topological ordering, changes when the graph evolves. This is necessary, since such a response might have already been returned by some process as the response of that operation.

We now introduce the Announce-Book-Commit Graph, a core component of our algorithm.

\subsection{Announce-Book-Commit (ABC) Graph}\label{sec:protocol:abc-graph}
Before introducing the main part of the algorithm, we describe the principal underlying data structure used to implement the dependency graph. 
The ABC Graph composes three data structures $A$, $B$ and $C$ into one snapshot object to allow an atomic read of all three structures.

\begin{enumerate}
    \item $A$ is a set of operations.
    \item $B$ is a set of tuples $(op,b)$ where $op$ is an operation and $b$ is an integer.
    \item $C$ is a graph of operations, encoded as a set of tuples $(op,M)$, where $op$ is an operation and $M$ is a set of operations, interpreted as the set of operations that have outgoing edges to $op$. 
\end{enumerate}

The ABC Graph exports three methods $\textit{add}_A(op)$, $\textit{add}_B(op, b)$ and $\textit{add}_C(op, M)$ used to append elements to the corresponding sets in the following way.
The set $A$ is the set of all operations $op$ with some preceding $\textit{add}_A(op)$ operation.
The set $B$ is the set of all tuples such that $(op,b) \in B$ if and only if there exists some preceding $\textit{add}_B(op,b)$ and $b$ is minimal among all preceding operations $\textit{add}_B(op,\cdot)$.
The graph $C$ is the set of all tuples such that $(op,M) \in C$ if and only if there exists some preceding operations $\textit{add}_C(op,M_j)$ and $M = \bigcup_{j} M_j$ for all preceding operations $\textit{add}_C(op,M_j)$.
The method $\textit{read}()$ returns the current state of the three data structures.
Note that, even though $\textit{add}_B(op,\cdot)$ and $\textit{add}_C(op,\cdot)$ may be called multiple times for the same operation $op$, the components $B$ and $C$, as returned by $\textit{read}()$, include at most one tuple per operation: the minimum value for $B$ and the union of all sets for $C$.

We say that an operation is \emph{announced} once it is added to $A$, \emph{booked} once it is added to $B$, and \emph{committed} once it is added to $C$.
Correspondingly, we say that a process is \emph{announcing}, \emph{booking}, or \emph{committing} an operation while executing $\textit{add}_A$, $\textit{add}_B$, or $\textit{add}_C$, respectively.
In particular, when we say that an operation $op$ is being \emph{committed for the first time}, we mean that some process is currently executing $\textit{add}_C(op,\cdot)$ that adds $op$ into $C$ for the first time.
We denote by $\textit{vertices}(B)$ (resp., $\textit{vertices}(C)$) the set of operations that appear in $B$ (resp., $C$).
Note that operations appear only as the first element of a tuple in $B$ but in both parts of a tuple in $C$.
Furthermore, we denote by $\textit{edges}(C)$ the set of edges of $C$, i.e., $(op,op') \in \textit{edges}(C)$ if and only if there exists $(op',M) \in C$ with $op \in M$.
The value associated with an operation $op$ in $B$ is called its \emph{$B$-value} and is denoted by $B[op]$.
If $op \not\in \textit{vertices}(B)$, we say that \emph{$op$ has no $B$-value}, and we define, by convention, $B[op]$ to be $\infty$. 

Given the specification of an ABC Graph, we now discuss its interpretation and use within our universal construction.
Each invocation of $\textit{publish}(op)$ announces its own operation $op$ by adding it to the set $A$, which thus contains every operation invoked so far.
For each operation $op$, the set $B$ records information about the recency of the ABC Graph that $\textit{publish}(op)$ has observed (namely, the size of $A$).
In our construction, each invocation $\textit{publish}(op)$ calls $\textit{add}_A(op)$ and $\textit{add}_B(op,\cdot)$ at most once, and only for its own parameter.
Consequently, the aggregation mechanism for $B$ exists solely for formal completeness and is not needed in practice.
Graph $C$ is interpreted as the dependency graph described above.

\begin{algorithm}[ht]
\caption{Announce-Book-Commit Graph}
\label{algo:abc-graph}
\footnotesize

\BlankLine
\BlankLine

\SetKwProg{Fn}{Local variables}{}{}
\Fn{}{
    $R_A \gets \{\:\}$, set of operations\;
    $R_B \gets \{\:\}$, set of tuples (operation, integer)\;
    $R_C \gets \{\:\}$, set of tuples (operation, set of operations)\;
}

\BlankLine
\BlankLine

\SetKwProg{Fn}{Shared objects}{}{}
\Fn{}{
    $S$: Snapshot object with $S_A[j], S_B[j], S_C[j]$ for process $p_j$\; \label{line:abc:snapshot-object}
}

\BlankLine
\BlankLine

\tcc{$op$: operation}
\SetKwProg{Fn}{Process $p_i$ upon $\boldsymbol{\textit{add}_A(op)}$}{}{}
\Fn{}{ \label{line:abc:addA}
    $R_A \gets R_A \cup \{op\}$\; \label{line:abc:addA:local}
    $S_A[i].\textit{write}(R_A)$\; \label{line:abc:addA:write}
}

\BlankLine
\BlankLine

\tcc{$op$: operation \newline $b$: integer}
\SetKwProg{Fn}{Process $p_i$ upon $\boldsymbol{\textit{add}_B(op, b)}$}{}{}
\Fn{}{ \label{line:abc:addB}
    $R_B \gets R_B \cup \{(op,b)\}$\; \label{line:abc:addB:local}
    $S_B[i].\textit{write}(R_B)$\; \label{line:abc:addB:write}
}

\BlankLine
\BlankLine

\tcc{$op$: operation \newline $M$: set of operations}
\SetKwProg{Fn}{Process $p_i$ upon $\boldsymbol{\textit{add}_C(op,M)}$}{}{}
\Fn{}{ \label{line:abc:addC}
    $R_C \gets R_C \cup \{(op,M)\}$\; \label{line:abc:addC:local}
    $S_C[i].\textit{write}(R_C)$\; \label{line:abc:addC:write}
}

\BlankLine
\BlankLine

\SetKwProg{Fn}{Process $p_i$ upon $\boldsymbol{\textit{read}()}$}{}{}
\Fn{}{ \label{line:abc:read}
    $T \gets S.\textit{snapshot}()$\; \label{line:abc:read:snapshot}
    \tcc{Group by $op$ and merge: \newline
    in $T_B$ - use $\min$ for integers \newline
    in $T_C$ - use $\cup$ for sets}
    \Return $\bigcup_{j} T_A[j],\ \bigcup_{j} T_B[j],\ \bigcup_{j} T_C[j]$\;
}

\BlankLine
\BlankLine

\end{algorithm}

Algorithm~\ref{algo:abc-graph} provides a wait-free and linearizable implementation of the ABC Graph. The proof is deferred to Appendix~\ref{sec:appendix:abc-graph}.
In the implementation, processes store invocations of $\textit{add}$ operations in local registers (lines~\ref{line:abc:addA:local},~\ref{line:abc:addB:local} and~\ref{line:abc:addC:local}) and write them to a shared snapshot object. The $\textit{read}$ operation aggregates all invocations of $\textit{add}$ to $A$, $B$ and $C$ as described above.

\subsection{Dynamically Concurrent Universal Construction}
In the following, we provide intuition for how the dynamically concurrent universal construction in Algorithm~\ref{algo:universal-construction} operates: how it places an operation $op$ correctly into the dependency graph, and why consensus is used only when necessary, satisfying dynamic concurrency (Definition~\ref{def:dynamic-concurrency}).
Formal proofs of linearizability, wait-freedom, and dynamic concurrency are deferred to the Appendices~\ref{sec:appendix:uc:linearizability},~\ref{sec:appendix:uc:wait-freedom}, and~\ref{sec:appendix:uc:dynamic-concurrency}.

\begin{algorithm}[ht]
\caption{Dynamically Concurrent Universal Construction}
\label{algo:universal-construction}
\footnotesize
\BlankLine
\BlankLine
\SetKwProg{Fn}{Shared objects}{}{}
\Fn{}{
    $K$: Snapshot object with $K[j]$ for process $p_j$, initially $[0,\dots,0]$\; \label{line:uc:snapshot-object}
    $G$: Announce-Book-Commit Graph\; \label{line:uc:abc-graph-object}
    $\textit{CONS}_j$: Consensus objects for $j \geq 1$\; \label{line:uc:consensus-objects}
}
\BlankLine
\BlankLine
\SetKwProg{Fn}{Process $p_i$ upon \textit{publish}($op$)}{}{} 
\Fn{}{\label{line:uc:propose}
    $G.\textit{add}_A(op)$\; \label{line:uc:addA}
    $(A,B,C) \gets G.\textit{read}()$\; \label{line:uc:read-1}
    $G.\textit{add}_B(op,|A|)$\; \label{line:uc:addB}
    $(A,B,C) \gets G.\textit{read}()$\; \label{line:uc:read-2}
    $l \gets \textit{linearize}(C)$\; \label{line:uc:linearize-C}
    \If{$op \in \textit{vertices}(C)$}{ \label{line:uc:checkC}
        \Return result of $op$ in $l$\; \label{line:uc:return-1}
    }
    \If{every subset of $A - \textit{vertices}(C) - \{op\}$ commutes with $op$ in $l$}{ \label{line:uc:commutativity}
        $G.\textit{add}_C(op,\textit{vertices}(C))$\; \label{line:uc:addC}
        \Return result of $op$ in $l \cdot (op)$\; \label{line:uc:return-2}
    }
    \tcc{ $\textit{CONFLICT RESOLUTION}$}
    $L \gets K.\textit{snapshot}()$\; \label{line:uc:rc:read-k}\label{line:uc:rc:start}
    $k \gets \max_j L[j]$\; \label{line:uc:rc:max-k}
    \While{$\textit{true}$}{ \label{line:uc:rc:loop}
        $k \gets k + 1$\; \label{line:uc:rc:inc-k}
        $(A,B,C) \gets G.\textit{read}()$\; \label{line:uc:rc:read-1}
        \If{$op \in C$}{ \label{line:uc:rc:check-before-return}
        \Return result of $op$ in $\textit{linearize}(C)$\;\label{line:uc:rc:return}
        }
        $\tilde{op} \in \arg\min_{op' \in B - \textit{vertices}(C)} B[op']$\; \label{line:uc:rc:select}
        $\hat{op} \gets \textit{CONS}_k.\textit{propose}(\tilde{op})$\; \label{line:uc:rc:consensus}
        $(A,B,C) \gets G.\textit{read}()$\; \label{line:uc:rc:read-2}
        \If{$\hat{op} \not\in \textit{vertices}(C)$}{ \label{line:uc:rc:check-before-addC}
            $G.\textit{add}_C(\hat{op},\textit{vertices}(C))$\; \label{line:uc:rc:addC}
        }
        $K[i].\textit{write}(k)$\; \label{line:uc:rc:write-k}\label{line:uc:rc:end}
    }
}
\BlankLine
\BlankLine
\end{algorithm}

Once a process proposes an operation $op$ (line~\ref{line:uc:propose}), it first adds it to set $A$ (line~\ref{line:uc:addA})
to make it visible to concurrent and future operations. 
Then it takes a snapshot (line~\ref{line:uc:read-1}) of the current state ($A$, $B$ and $C$). 
By construction, graph $C$ contains the previously committed operations. 
The operations in $A-\textit{vertices}(C)$ are therefore concurrent. 
Let us call this set $\textit{CUR}$. 
Note that some operations in $\textit{CUR}$ may never be committed, as they may be issued by processes that have later crashed.

After reading the state for the first time, $op$ is added to $B$ (line~\ref{line:uc:addB}) together with $|A|$, the cardinality of $A$.
Intuitively, $|A|$ indicates how many operations $\textit{publish}(op)$ has encountered at its first read, indicating the recency of its view of the system. This helps later to prioritize operations that may have missed conflicts.
After announcing and booking the operation, the process reads the state again (line~\ref{line:uc:read-2}), which updates its local values for $A$, $B$ and $C$ (and implicitly $\textit{CUR}$).
The graph $C$ is then topologically ordered (line $\ref{line:uc:linearize-C}$). 
We denote this ordering by $l$.
We assume inductively that all topological orderings of $C$ are equivalent.
If $op$ is already committed (line~\ref{line:uc:checkC}), the result of the operation in $l$ (according to the sequential specification) is returned (line~\ref{line:uc:return-1}); otherwise, the process checks whether $op$ commutes with every subset of $\textit{CUR} - \{op\}$ in $l$ (line~\ref{line:uc:commutativity}).
Verifying commutativity with \emph{every} subset of $\textit{CUR} - \{op\}$ accounts for the possibility that some operations in $\textit{CUR}$ may never be committed due to crash failures.
Moreover, Lemma~\ref{lemma:equivalence-preservation} guarantees that, for every equivalent ordering of $\textit{CUR} - \{op\}$ in $l$, every possible insertion point of $op$ within those orderings is equivalent in $l$, as formalized in Definition~\ref{def:equivalence-relation-of-orderings}.
If commutativity is verified, $\textit{publish}(op)$ commits $op$ (line~\ref{line:uc:addC}) and returns its response in $l \cdot (op)$ (line~\ref{line:uc:return-2}).
Otherwise, ordering of concurrent operations is required and conflict resolution (lines~\ref{line:uc:rc:start}-\ref{line:uc:rc:end}), where consensus is used, is executed.

Conflict resolution proceeds in iterations. 
In each iteration, the process reads the state (line~\ref{line:uc:rc:read-1}), selects an operation with the minimal $B$-value (among all non-committed operations) (line~\ref{line:uc:rc:select}) and proposes it to the next consensus object (line~\ref{line:uc:rc:consensus}).
The state is read again (line~\ref{line:uc:rc:read-2}) and unless the operation $\hat{op}$ that has been returned by the consensus object is already committed (line~\ref{line:uc:rc:check-before-addC}), it is committed by $\textit{publish}(op)$ in line~\ref{line:uc:rc:addC}.

Each operation $\tilde{op}$ that is proposed to a consensus object, and therefore each operation that is committed in conflict resolution, has minimal $B$-value, according to a version of the ABC Graph that was read by some process executing conflict resolution.
In the following, we are interested in the $B$-values according to this version of the ABC Graph.
As can be easily verified, an operation is added to $B$ at most once, and only by its corresponding invocation of $\textit{publish}$.
Consequently, the $B$-values of different versions of the ABC Graph differ only in the set of operations that have a (finite) $B$-value.
Once an operation has a (finite) $B$-value, it is constant.
Since we consider, by construction, a version of the ABC Graph where $\tilde{op}$ has a minimal $B$-value (among all non-committed operations), there only exist operations with a higher $B$-value or operations that do not have a $B$ value yet.
Each operation that already has, or can obtain, a $B$-value of at most that of $\tilde{op}$ has been encountered by $\textit{publish}(\tilde{op})$ in line~\ref{line:uc:read-2}.\footnote{Since the ABC Graph is linearizable, we can order all invocations of $\textit{add}_A$. Let $op_i$ be the $i$-th operation added to $A$. Since the ABC Graph is read again after executing $\textit{add}_A(op_i)$, $op_i$ will have a $B$-value of at least $i$, once it is booked. An operation that already has, or can obtain, a $B$-value of at most $i$, has to be among the first $i$ operations added to $A$, and has therefore been encountered in line~\ref{line:uc:read-2} by an operation with a $B$-value of $i$.}

One can notice that $\textit{publish}(\tilde{op})$ might have not encountered (in line~\ref{line:uc:read-2}) concurrent operations with higher (possibly infinite) $B$-values.
Therefore, $\textit{publish}(\tilde{op})$ might have encountered no conflict and already passed the check of commutativity (line~\ref{line:uc:commutativity}). 
This is the reason why operations with minimal $B$-value are prioritized in conflict resolution.
Operations with no $B$-value, i.e., operations that are not booked yet, will reread the state (line~\ref{line:uc:read-2}) before committing and therefore encounter every operation that already has a $B$-value.
Thus, an operation with no $B$-value will either encounter a conflict too or find every conflict already resolved in $C$.
Therefore, operations that are not in $B$, i.e. with no $B$-value, need not to be prioritized in the conflict resolution.
In summary, an operation with minimal $B$-value (among concurrent operations) in some version of the ABC Graph has been or will be encountered in line~\ref{line:uc:read-2} by $\textit{publish}(op^*)$ for every non-committed operation $op^*$. Such an operation can therefore be proposed to be committed in conflict resolution without risking conflicting operations to miss it.

Clearly, when an operation is proposed to $\textit{CONS}_k$, every previous $\textit{CONS}_m$ with $m<k$ stores an operation and those operations are committed. 
This does not imply that an operation stored in $\textit{CONS}_k$, once committed, will include all of those operations, stored in previous consensus objects, as dependencies. This is due to the fact, that operations can be concurrently committed in the conflict-free path, while being proposed to a consensus object.
However, it is easy to see that if an operation $op$ is committed in the conflict resolution path, every operation \emph{committed in the conflict resolution path} using a later consensus object will include $op$ in its dependencies.

The algorithm uses a snapshot object (line~\ref{line:uc:snapshot-object}) that records the last completed iteration of every process, allowing processes to skip iterations that have already been completed by any process (lines~\ref{line:uc:rc:read-k},~\ref{line:uc:rc:max-k},~\ref{line:uc:rc:inc-k} and~\ref{line:uc:rc:write-k}). This does not affect correctness, as re-executing a completed iteration has no effect.

\myparagraph{Sketch of Linearizability Proof} Algorithm~\ref{algo:universal-construction} ensures linearizability by checking commutativity among operations and resorting to consensus when conflicts arise. No two conflicting operations can proceed along the conflict-free path simultaneously, because at least one operation will detect the other. Operations that could have ``missed'' conflicts are prioritized, when committing operations in the conflict resolution path, to maintain a consistent view among all processes. While a conflict is being resolved, other conflicting operations also enter conflict resolution, whereas non-conflicting operations continue along the conflict-free path. The resulting dependency graph has topological orderings that are equivalent, and each such ordering serves as a linearization of the history.

\myparagraph{Sketch of Wait-Freedom Proof} Under the assumption that the ABC Graph is wait-free, any invocation of $\textit{publish}$ that encounters no conflict terminates trivially after a finite number of steps. For invocations that perform conflict resolution, the process skips iterations that have already been completed. Given at most $c \leq n$ concurrent processes, the corresponding operation is stored in a consensus object after at most $c+1$ iterations, unless it is concurrently committed via the conflict-free path. In all cases, $\textit{publish}$ completes after a finite number of steps.

\myparagraph{Sketch of Dynamic Concurrency Proof} When reading the ABC Graph (line~\ref{line:uc:read-2}), the tuple $\big(\textit{linearize}(C),A-\textit{vertices}(C)-\{op\}\big)$ essentially represents a system state.
It can be shown that $l := \textit{linearize}(C)$ is a linearization of the history induced by the current run, and that for every possible extension of the run the induced history has a linearization such that $l$ is a prefix of that linearization.
Furthermore, $A-\textit{vertices}(C)-\{op\}$ consists of concurrent operations $op'$ such that $op' \not\in \textit{linearize}(C)$ and $op' \neq op$, as required by Definition~\ref{def:system-state}.
Note that $A-\textit{vertices}(C)-\{op\}$ might be a subset of all concurrent operations, which only restricts the cases where consensus is invoked. Therefore, the check performed in line~\ref{line:uc:commutativity} directly corresponds to the dynamic concurrency condition: consensus is used only when the operation does not commute with some subset of the concurrent operations in that system state.

As mentioned above, formal proofs of linearizability, wait-freedom, and dynamic concurrency are provided in the Appendices~\ref{sec:appendix:uc:linearizability},~\ref{sec:appendix:uc:wait-freedom}, and~\ref{sec:appendix:uc:dynamic-concurrency}.

\section{Discussion}
\label{sec:conc}
In this paper, we settle down the very \emph{possibility} of implementing a dynamically concurrent universal construction that only resorts to consensus if the concurrent operations are in conflict, given the current object state.

The next question is what the \emph{cost} of detecting and resolving dynamic conflicts is.
In our construction, non-conflicting operations complete in two write-read rounds.
However, both the communication overhead of these rounds and the local computation can become significant.
The latter is exponential in the number of concurrent operations, as commutativity must be verified for every subset of those.
In contrast, consensus itself may not be particularly costly, especially in a shared-memory setting, which suggests a trade-off between minimizing synchronization through dynamic conflict detection and the total computational and communication costs associated with it.
It may well be that these costs outweigh the benefits gained from avoiding consensus in conflict-free scenarios.
It is therefore appealing to undertake comparative performance analysis of universal-construction algorithms, with and without dynamic conflict resolution.

It is important to note that our construction makes no assumptions about the type of operations. 
It is common to avoid synchronization in \emph{read-only} workloads, and the same optimization is possible here. 
Read-only operations do not need to be added to the dependency graph, since other processes do not need to be aware of them. 
They can compute their responses directly based on an up-to-date version of the dependency graph, proceeding efficiently while preserving correctness.

The current definition of a concurrent operation accounts for \emph{interval contention}~\cite{solofast}---operations are considered concurrent if their intervals overlap.
A more refined definition accounts for \emph{step contention} which only accounts for \emph{concurrently active} (taking steps) operations~\cite{HLM03-of,solofast}.
It is tempting to design a universal construction that uses consensus only when a concurrently active conflicting operation is detected.
We conjecture, however, that this might not be possible.\footnote{
Note that every sequential objects admits a \emph{solo-fast} implementation~\cite{solofast}, in which operations complete without using strong synchronization primitives in the absence of step contention.
However, using a solo-fast consensus object implementation in our algorithm would not achieve the desired effect.
Even if no conflicting operation takes steps concurrently, some commuting operations could still introduce step contention on the solo-fast consensus implementation.
}

This work is close in spirit to the idea of \emph{optimal concurrency}~\cite{oc-sirocco,oc-pact}, an earlier attempt to grasp the \emph{minimal} use of synchronization in concurrent data structures.
By na\"ively running the \emph{sequential} implementation of an object in the concurrent environment, we obtain, intuitively, the throughput baseline: no implementation is likely to perform better~\cite{DGT15}. 
We can turn this (most probably incorrect) sequential implementation into a correct concurrent one by introducing synchronization primitives. 
One can then reason about \emph{the amount of concurrency} the implementation exhibits using  the set of correct (linearizable) concurrent \emph{schedules} it accepts.
An \emph{optimally concurrent} implementation~\cite{oc-sirocco} rejects a schedule only if it results in an  \emph{incorrect} (i.e., not linearizable) execution.
Examples of optimally concurrent implementations of concrete data structures are given in~\cite{oc-sirocco,oc-pact}.
In contrast, in this paper, we rely on the sequential \emph{specification} of the object, rather then on a concrete sequential implementation of it. 
This makes our approach more general, even though we do not argue that the use of consensus is minimal in any sense.
It is intriguing to explore the issue of ``optimality'' further, e.g., can one argue that no algorithm can use less strong synchronization? 

Finally, dynamic conflict resolution can become an important tool in Byzantine fault-tolerant universal constructions~\cite{bitcoin,ethereum}, in particular in DAG-based blockchains~\cite{dagrider,bullshark,shoal}.

%\clearpage

\bibliographystyle{abbrv}
\bibliography{references}

\clearpage

\appendix

\section{ABC Graph - Proof of Linearizability and Wait-freedom}
\label{sec:appendix:abc-graph}
In this section, we prove in Lemma~\ref{lemma:abc-graph:linearizability} that the implementation of the ABC Graph given in Algorithm~\ref{algo:abc-graph} is both linearizable and wait-free, assuming that the underlying snapshot object $S$ (line~\ref{line:abc:snapshot-object}) itself is linearizable and wait-free.

\begin{lemma}
    \label{lemma:abc-graph:linearizability}
    Algorithm~\ref{algo:abc-graph} is a wait-free linearizable implementation of an ABC Graph.
\end{lemma}
\begin{proof}
Let $R$ be a run of Algorithm~\ref{algo:abc-graph}, and let $H$ be the history induced by $R$ with respect to the ABC Graph. Furthermore, let $H_S$ denote the history induced by $R$ with respect to the shared snapshot object $S$, i.e., $H$ contains the invocations and responses of ABC Graph operations, while $H_S$ contains the invocations and responses of snapshot operations. Each ABC Graph operation accesses $S$ exactly once (lines~\ref{line:abc:addA:write},~\ref{line:abc:addB:write},~\ref{line:abc:addC:write},~\ref{line:abc:read:snapshot}). 
Consequently, there is a unique mapping from operations of $S$ to the corresponding operations of the ABC Graph (lines~\ref{line:abc:addA},~\ref{line:abc:addB},~\ref{line:abc:addC},~\ref{line:abc:read}).

Let $L_S$ be a linearization of $H_S$, and let $L$ be the sequence of ABC Graph operations obtained by mapping each operation in $L_S$ to its originating ABC Graph operation. We show that $L$ is a linearization of $H$.

For each $op.res \in H$ there is a corresponding operation of $S$ in $L_S$ and therefore $op \in L$. Similarly, for each operation on $S$ in $L_S$ there is a corresponding $op$ with $op.inv \in H$. Assume that there are $op,op' \in H$ with $op <_H op'$. It follows, that the corresponding operation on $S$ of $op$ precedes the one of $op'$ in $H_S$ and therefore in $L_S$ and consequently $op <_L op'$. We conclude that $L$ is a sequential history that is equivalent to a completion of $H$ and preserves the precedence relation of $H$.

It remains to show that $L$ is legal. Among the ABC Graph operations, only $\textit{read}()$ (line~\ref{line:abc:read}) produces a response value. This value is determined by the response value of $S.\textit{snapshot}()$ in line~\ref{line:abc:read:snapshot}. Since $S$ is linearizable, its response value contains exactly those $\textit{write}$ operations that precede it in $L_S$. Consequently, the response value of $\textit{read}()$ contains precisely the $\textit{add}_X$ operations ($X \in \{A,B,C\}$) that precede it in $L$. The aggregation of $A$, $B$, and $C$ coincides with the specification given in Section~\ref{sec:protocol:abc-graph}. Therefore, Algorithm~\ref{algo:abc-graph} is a linearizable implementation of an ABC Graph.

Algorithm~\ref{algo:abc-graph} uses a single shared object: the snapshot object $S$ (line~\ref{line:abc:snapshot-object}). Each ABC Graph operation performs exactly one operation on $S$ along with some local computation. As a result, the number of steps required by the ABC Graph implementations matches that of the underlying snapshot object. Therefore, assuming $S$ is wait-free, the ABC Graph implementation is wait-free as well.
\end{proof}

\section{Universal Construction - Proof of Linearizability}
\label{sec:appendix:uc:linearizability}
In the proof of correctness for Algorithm~\ref{algo:universal-construction}, we assume discrete global time, inaccessible to any individual process. 
Since Algorithm~\ref{algo:abc-graph} is a linearizable implementation of the ABC Graph (Lemma~\ref{lemma:abc-graph:linearizability}), the execution of operations on this object admits a well-defined total order.
We start from an initial time $t=0$. Each invocation of $\textit{add}_A$, $\textit{add}_B$ or $\textit{add}_C$ produces an \emph{atomic extension}, which transforms the graph at time $t$ into the graph at time $t+1$.
We denote such an atomic extension by $t \rightarrow t+1$.
The causing $\textit{add}$ operation is assigned the timestamp $t+1$. An \emph{extension} is then defined as a sequence of one or more consecutive atomic extensions, representing the evolution of the graph across time.
Since the $\textit{read}$ method is state-preserving, we assign it the same timestamp as its most recent predecessor (or $t=0$ if none exists). 
Consequently, if the same timestamp $t$ is assigned to a $\textit{read}$ and an $\textit{add}$ operation, the $\textit{read}$ observes the new state, i.e., is scheduled after the $\textit{add}$ operation. 
We denote $A_t$, $B_t$ and $C_t$ as the values of $A$, $B$ and $C$ at time $t$.

As mentioned at the beginning of Section~\ref{sec:protocol}, linearizability of our construction is implied by the properties introduced there, which are consolidated and slightly extended in their formal statement in Lemma~\ref{lemma:correctness}.
The remainder of this section is dedicated to proving these properties and concluding in Lemma~\ref{lemma:linearizability} that the universal construction in Algorithm~\ref{algo:universal-construction} is linearizable.

\begin{lemma}
    \label{lemma:correctness}
    Each run of Algorithm~\ref{algo:universal-construction} fulfills the following properties. Let $op$ and $op'$ be any two operations in that run.
    \begin{itemize}
        \item For all times $t \geq 0$, $C_t$ is acyclic.
        \item Upon termination of $\textit{publish}(op)$, the operation $op$ is committed.
        \item If $\textit{publish}(op)$ precedes $\textit{publish}(op')$, then $op$ precedes $op'$ in every topological ordering of $C_t$ for every $t$ after $op'$ has been committed.
        \item Let $t_1, t_2 \geq 0$, and let $s_1$ and $s_2$ be topological orderings of $C_{t_1}$ and $C_{t_2}$, respectively.
        For every operation $op$ such that $op \in s_1$ and $op \in s_2$, the sequential specification yields the same response for $op$ in both $s_1$ and $s_2$.
        \item For every time $t_2 \geq 0$ and every operation $op \in C_{t_2}$, if $op$ has a response, then there exists a time $t_1$ with $0 \leq t_1 \leq t_2$ and a topological ordering $T$ of $C_{t_1}$ that contains $op$, such that the response of $op$ agrees with the sequential specification applied to $T$.
    \end{itemize}
\end{lemma}

\subsection{Fundamental Properties and Acyclicity}

After expanding the notion of the inclusion relation to graphs in Definition~\ref{def:graph-inclusion}, we establish some basic properties of Algorithm~\ref{algo:universal-construction} through a series of Lemmata. 
Important insights are for example that the ABC Graph is only growing (Lemma~\ref{lemma:basic-properties}), only growing by at most one node at a time (Lemma~\ref{lemma:previous-nodes}) and that the graph $C$ is transitive (Lemma~\ref{lemma:transitivity}) and acyclic (Lemma~\ref{lemma:acyclicity}) and consequently suitable for topological ordering.
Lemma~\ref{lemma:separation-of-orderings} shows that it is always possible to select a topological ordering of a version of $C$ in which a prefix forms a topological ordering of an earlier version of $C$.

\begin{definition}
    \label{def:graph-inclusion}
    Let $G$ and $H$ be graphs. We write $G \subseteq H$ if $\textit{vertices}(G) \subseteq \textit{vertices}(H) \wedge \textit{edges}(G) \subseteq \textit{edges}(H)$.
\end{definition}

\begin{lemma}
    \label{lemma:basic-properties}
    The sets $A_t$ and $B_t$, and graph $C_t$ satisfy the following properties for every time $t \geq 0$:
    \begin{enumerate}
        \item $A_t \subseteq \textit{vertices}(B_t) \subseteq \textit{vertices}(C_t)$.
        \label{lemma:basic-properties.subset}
        \item Let $op$ be an operation. No other process than the one executing $\textit{publish}(op)$ announces or books $op$. It is announced at most once and booked at most once.
        \label{lemma:basic-properties.A-and-B-once}
        \item If $t \leq t'$, then $A_t \subseteq A_{t'}$, $B_t \subseteq B_{t'}$, and $C_t \subseteq C_{t'}$.
        \label{lemma:basic-properties.increasing}
        \item For any finite set of times $t_1 \leq \dots \leq t_m$, we have $C_{t_m} = \bigcup_j C_{t_j}$.
        \label{lemma:basic-properties.set-of-graphs}
    \end{enumerate}
\end{lemma}
\begin{proof}
    Items \ref{lemma:basic-properties.subset}-\ref{lemma:basic-properties.increasing} follow trivially from the implementation.
    Item~\ref{lemma:basic-properties.set-of-graphs} follows directly from item~\ref{lemma:basic-properties.increasing}.
\end{proof}

\begin{lemma}
    \label{lemma:irreflexivity}
    The graph $C_t$ is irreflexive for every time $t \geq 0$.
\end{lemma}
\begin{proof}
    The initial graph $C_{0}$ is empty and trivially irreflexive. Let $t \rightarrow t+1$ be an atomic extension of the ABC Graph. If this is the result of an invocation of $\textit{add}_A$ or $\textit{add}_B$, we have $C_{t+1} = C_t$ and irreflexivity follows by induction. 
    Assume that $t \rightarrow t+1$ is the result of an invocation of $\textit{add}_C(op,M)$ and therefore adds incoming edges to at most one node. 
    To prove irreflexivity, we have to show that $op \notin M$.
    The method $\textit{add}_C(op,M)$ is called in lines \ref{line:uc:addC} and \ref{line:uc:rc:addC}. 
    In both cases, the condition $op \notin M$ has been explicitly checked previously.
\end{proof}

\begin{lemma}
    \label{lemma:previous-nodes}
    Let $E_t := t \rightarrow t+1$ be an atomic extension of the ABC Graph and let this extension be the result of an invocation of $\textit{add}_C(op,M)$. We have
    \begin{enumerate}
        \item $M = \textit{vertices}(C_{t'})$ for some time $t' \leq t$ and
        \item $\textit{vertices}(C_{t+1}) - \textit{vertices}(C_t) \subseteq \{op\}$.
    \end{enumerate}
    We define $C(E_t) := C_{t'}$ for the $C_{t'}$ that was used to calculate $M$.
\end{lemma}
\begin{proof}
    The method $\textit{add}_C(op,M)$ is called in lines \ref{line:uc:addC} and \ref{line:uc:rc:addC}. In both cases it is called with $M=\textit{vertices}(C')$ for some previously read graph $C'$. With Lemma~\ref{lemma:basic-properties}~(\ref{lemma:basic-properties.increasing}) we have $M \subseteq \textit{vertices}(C_t)$ and therefore $\textit{vertices}(C_{t+1}) - \textit{vertices}(C_t) \subseteq \{op\}$. As multiple graphs $C_{t'}$ might fulfill $M=\textit{vertices}(C_{t'})$, we define $C(E_t)$ as the $C_{t'}$ that was used in the actual execution of the algorithm to calculate $M$.
\end{proof}

Lemma~\ref{lemma:previous-nodes} states that, in any atomic extension, caused by an invocation of $\textit{add}_C(op,M)$, in which incoming edges are added from the nodes in $M$ to the node $op$, the set $M$ must correspond to the nodes of a previous version of graph $C$. Consequently, at most one operation (namely, $op$) is added to graph $C$ in such an extension, since every operation in $M$ is already present in the graph.
Furthermore, we can infer that for every time $t \geq 0$ and operation $op \in C_t$, the exist times $t_1 \leq \dots \leq t_k$ such that the set of nodes that have incoming edges to $op$ in $C_t$ is the union of all $\textit{vertices}(C_{t_i})$ with $1 \leq i \leq k$, which equals $\textit{vertices}(C_{t_k})$ by Lemma~\ref{lemma:basic-properties}~(\ref{lemma:basic-properties.set-of-graphs}).

\begin{lemma}
    \label{lemma:transitivity}
    The graph $C_t$ is transitive for every time $t \geq 0$.
\end{lemma}
\begin{proof}
    Proof by contradiction. Assume $C_t$ is not transitive. There exists $a,b,c \in \textit{vertices}(C_t)$ with $(a,b),(b,c) \in \textit{edges}(C_t)$ but $(a,c) \notin \textit{edges}(C_t)$. Let us denote the first time $a$ was added to $C$ as $t_a$, and the first time $b$ was added to $C$ as $t_b$ (i.e., we have $b \notin \textit{vertices}(C_{t'})$ for all $t'<t_b$ and $b \in \textit{vertices}(C_{t'})$ for all $t' \geq t_b$). As we have $(b,c) \in \textit{edges}(C_t)$, there must exist a time $t_{\textit{read}} \geq t_b$ such that $C_{t_{\textit{read}}}$ was used to add $c$ (or incoming edges to $c$) to $C$. With $(a,c) \notin \textit{edges}(C_t)$, we conclude $t_b \leq t_{\textit{read}} < t_a$. As we have $(a,b) \in \textit{edges}(C_t)$ there must have been a read of $C$ after (or at the same time as) $t_a$ that was used to add incoming edges to $b$. Therefore, this read is after $t_b$ and must have read $b$. As $C_t$ is irreflexive by Lemma \ref{lemma:irreflexivity}, this is a contradiction.
\end{proof}

\begin{lemma}
    \label{lemma:acyclicity}
    The graph $C_t$ is acyclic for every time $t \geq 0$.
\end{lemma}
\begin{proof}
    The graph $C_t$ is irreflexive by Lemma \ref{lemma:irreflexivity} and transitive by Lemma \ref{lemma:transitivity}. Acyclicity follows directly.
\end{proof}

\begin{lemma}
    \label{lemma:no-backwards-edge}
    Let $0 \leq t_1 \leq t_2$, $a \in \textit{vertices}(C_{t_1})$ and $b \in \textit{vertices}(C_{t_2}) - \textit{vertices}(C_{t_1})$.
    Then for all $t \geq 0$, $(b,a) \notin C_t$.
\end{lemma}
\begin{proof}
    The claim is trivial for $t_1 = t_2$. Assume $t_1 < t_2$.
    We proceed by contradiction. 
    Suppose that there exists $a \in \textit{vertices}(C_{t_1})$ and $b \in \textit{vertices}(C_{t_2}) - \textit{vertices}(C_{t_1})$, such that, for some time $t' \geq 0$, $(b,a) \in C_{t'}$. 
    Let us call $t_a$ the first time when $a$ was committed and $t_b$ the first time when $b$ was committed.
    Since $a \in \textit{vertices}(C_{t_1})$, it follows that $t_a \leq t_1$ and since $b \in \textit{vertices}(C_{t_2}) - \textit{vertices}(C_{t_1})$, it follows that $t_1 < t_b$ and therefore $t_a < t_b$. 
    Since $(b,a) \in C_{t'}$, it follows that $a$ was committed with $C_{t_{\textit{read}}}$ as dependency with $t_b \leq t_{\textit{read}}$. We conclude $t_a < t_{\textit{read}}$, which contradicts acyclicity proven in Lemma~\ref{lemma:acyclicity}.
\end{proof}

\begin{lemma}
    \label{lemma:separation-of-orderings}
    Let $0 \leq t_1 \leq t_2$. 
    There exists a topological ordering $r$ of $C_{t_2}$ such that a prefix of $r$ is a topological ordering of $C_{t_1}$.
\end{lemma}
\begin{proof}
    By Lemma~\ref{lemma:acyclicity}, there exists a topological ordering $s$ of $C_{t_2}$.
    By Lemma~\ref{lemma:no-backwards-edge}, there are no edges in $C_{t_2}$ from nodes in $\textit{vertices}(C_{t_2}) - \textit{vertices}(C_{t_1})$ to those in $\textit{vertices}(C_{t_1})$.
    Consequently, if we define $r := r_1 \cdot r_2$, where $r_1$ is the subsequence of $s$ with exactly the nodes in $C_{t_1}$ and $r_2$ is the subsequence of $s$ with exactly the nodes in $\textit{vertices}(C_{t_2}) - \textit{vertices}(C_{t_1})$, then $r$ forms a topological ordering of $C_{t_2}$.
    By Lemma~\ref{lemma:basic-properties}~(\ref{lemma:basic-properties.increasing}), $r_1$ is a topological ordering of $C_{t_1}$.
\end{proof}

\subsection{Commitment and Real-time Ordering}

Having established fundamental properties of the ABC Graph when used by the Algorithm~\ref{algo:universal-construction}, most notably acyclicity and, consequently, the existence of topological orderings, in this section we proceed to prove in Lemma~\ref{lemma:commitment} that an operation is committed before its corresponding invocation of the $publish$ method terminates, and in Lemma~\ref{lemma:precedence} that the algorithm satisfies real-time ordering.

\begin{lemma}
    \label{lemma:commitment}
    Upon termination of $\textit{publish}(op)$, the operation $op$ is committed.
\end{lemma}
\begin{proof}
    The method $\textit{publish}$ can return at lines~\ref{line:uc:return-1},~\ref{line:uc:return-2} and~\ref{line:uc:rc:return}.
    If $\textit{publish}(op)$ returned at line~\ref{line:uc:return-1}, it had verified that $op$ is committed in line~\ref{line:uc:checkC}.
    If $\textit{publish}(op)$ returned at line~\ref{line:uc:return-2}, it committed $op$ in line~\ref{line:uc:addC}.
    If $\textit{publish}(op)$ returned at line~\ref{line:uc:rc:return}, it had verified that $op$ is committed in line~\ref{line:uc:rc:check-before-return}.
    Consequently, $op$ is committed at return.
\end{proof}

\begin{lemma}
    \label{lemma:precedence}
    If $\textit{publish}(op)$ precedes $\textit{publish}(op')$, then $op$ precedes $op'$ in every topological ordering of $C_t$ for every $t$ after $op'$ has been committed.
\end{lemma}
\begin{proof}
    Suppose $\textit{publish}(op)$ precedes $\textit{publish}(op')$. By Lemma~\ref{lemma:commitment}, $op$ is committed upon the termination of $\textit{publish}(op)$. By Lemma~\ref{lemma:previous-nodes}, there exists a graph $C_{t'}$ whose nodes have been used as incoming edges when committing $op'$.
    Suppose that this graph $C_{t'}$ has been read after the termination of $\textit{publish}(op)$.
    We therefore have $op \in C_{t'}$, implying that the edge $(op,op')$ is added to $C$ by the time $\textit{publish}(op')$ terminates.
    The claim then follows directly.
    
    It remains to show that $C_{t'}$ was read after the termination of $\textit{publish}(op)$.
    If $op'$ was committed by $publish(op')$, $C_{t'}$ was read during the execution of $\textit{publish}(op')$ and therefore after the termination of $\textit{publish}(op)$.
    Otherwise, suppose that $op'$ was committed by another invocation $p$ of $\textit{publish}$ and thus in conflict resolution.
    The graph $C_{t'}$ was read by $p$ in line~\ref{line:uc:rc:read-2} after some invocation of $\textit{publish}$ had proposed $op'$ to some consensus object.
    This occurred after $op'$ had been booked by $\textit{publish}(op')$ (Lemma~\ref{lemma:basic-properties}~(\ref{lemma:basic-properties.A-and-B-once})) and therefore after the termination of $\textit{publish(op)}$.
\end{proof}

\subsection{Equivalence of Topological Orderings}
After proving the first three claims of Lemma~\ref{lemma:correctness}, it remains to establish the final two. 
Before returning to these claims in Section~\ref{sec:appendix:uc:linearizability:proof-completion}, we first prove Lemma~\ref{lemma:equivalence-of-every-topological-ordering} as an intermediate result.

Recall the definition of the graph $C(E_t)$ from Lemma~\ref{lemma:previous-nodes} for an atomic extension $E_t$ that is the result of an invocation of $\textit{add}_C(op,M)$. We have $M = \textit{vertices}(C(E_t))$.
If $op$ was added to $C$ for the first time at $E_t$, $C(E_t)$ consists precisely of those operations of $C_t$ (equivalently, of $C_{t+1}$) that have an outgoing edge to $op$ in $C_{t+1}$.

\begin{definition}
    \label{def:P(E_t)}
    Let $E_t := t \rightarrow t+1$ be an atomic extension of the ABC Graph and let this extension be the result of an invocation of $\textit{add}_C(op,M)$. Let $C(E_t)$ be defined as in Lemma~\ref{lemma:previous-nodes}, such that $M = \textit{vertices}(C(E_t))$.
    Furthermore, let $op \not\in C_t$.
    
    We define $P(E_t) \subseteq C_t$ as the subgraph that is disjoint from $C(E_t)$.
\end{definition}

Given $P(E_t)$ as defined in Definition~\ref{def:P(E_t)}, for all $op' \in C_{t+1}$, exactly one of the following holds: $op' \in C(E_t)$, $op' \in P(E_t)$, or $op' = op$.
Note that we only define $P(E_t)$ in the case of $op \notin C_t$, i.e., $op$ is committed for the first time.
Intuitively, $P(E_t)$ captures the portion of the graph that was added after $C(E_t)$ but before $op$ was added to $C_t$.\footnote{The extension(s) from $C(E_t)$ to $C_t$ may include additional edges within $C(E_t)$, as well as edges from nodes in $C(E_t)$ to nodes in $P(E_t)$. This information is not represented when partitioning $C_t$ in $C(E_t)$ and $P(E_t)$.} Although edges from nodes in $C(E_t)$ to nodes in $P(E_t)$ may exist, the reverse is ruled out by transitivity (Lemma~\ref{lemma:transitivity}) and the definition of $P(E_t)$: if a node in $P(E_t)$ had an edge to a node in $C(E_t)$, it would (by transitivity) also have an edge to $op$ and thus would belong to $C(E_t)$, contradicting disjointness.\footnote{Note that $op \notin C_t$ by assumption and therefore $(op',op) \in C_{t+1}$ implies $op' \in C(E_t)$ for every $op' \in C_{t+1}$.}

\begin{definition}
    \label{def:graph-cap-sequence}
    Let $G$ be a graph whose nodes are operations, and let $s$ be a sequence of operations. We write $G \cap s$ for the subgraph of $G$ induced by the set of nodes appearing in $s$.
\end{definition}

\begin{lemma}
    \label{lemma:equivalence-of-every-topological-ordering}
    For each fixed $t \geq 0$, all topological orderings of $C_t$ are equivalent. Moreover, this equivalence holds not only for the complete sequences but also for any pair of prefixes consisting of the same set of operations.
\end{lemma}
\begin{proof}

Let $C_0$, $\dots$, $C_t$, $C_{t+1}$, $\dots$ be the states of graph $C$ as it evolves over time. 
As the induction hypothesis, assume that for all $t' \leq t$, any two topological orderings of $C_{t'}$ are equivalent - not only as full sequences, but also in the sense that any pair of prefixes consisting of the same set of operations are themselves equivalent.

As the hypothesis holds trivially for the empty graph $C_0$, it remains to show that this property also holds for $C_{t+1}$. 
If the atomic extension $E_t := t \rightarrow t+1$ is the result of an invocation of $\textit{add}_A$ or $\textit{add}_B$, we have $C_{t+1} = C_t$ and the statement follows directly from the induction hypothesis.
If only edges are added during $E_t$, the statement also follows directly from the induction hypothesis, as every topological ordering of $C_{t+1}$ is also a valid topological ordering of $C_t$.

Now suppose that a new operation $op$ is committed (added to $C$) for the first time in the atomic extension $E_t$. 
Let us consider two prefixes of topological orderings of $C_{t+1}$ that consist of the same set of operations. If $op$ is not among these operations, the equivalence follows from the induction hypothesis. In the following, we assume that $op$ is included in both prefixes.

We define two subgraphs to analyze the structural implications of adding $op$ to $C_t$. 
First, let $C(E_t)$, as defined in Lemma~\ref{lemma:previous-nodes}, denote the subgraph that has been read when $op$ was added - that is, every node in $C(E_t)$ has an outgoing edge to $op$. Note that, since $op$ was newly added, only nodes in $C(E_t)$ have edges to it in $C_{t+1}$.
As every node in $C(E_t)$ has an outgoing edge to $op$, the nodes in $C(E_t)$ must appear before $op$ in any topological ordering of $C_{t+1}$. 
Therefore, if a prefix of a topological ordering includes $op$, it must also include all nodes of $C(E_t)$. Hence, the presence of $op$ in both prefixes implies that the nodes of $C(E_t)$ are also included in both prefixes.
Second, let $P(E_t) \subseteq C_t$, as defined in Definition~\ref{def:P(E_t)}, denote the subgraph that is disjoint from $C(E_t)$.
As $op\notin C_t$,  $P(E_t)$ does not contain $op$.

Let $o = o_0 \cdot op \cdot o_1$ be a prefix of a topological ordering of $C_{t+1}$. Since every node in $C(E_t)$ has an edge to $op$, it follows that $\textit{vertices}(C(E_t)) \subseteq o_0$. By the induction hypothesis, we have $o_0 \simeq c \cdot \tilde{o_0}$, where $c$ is any topological ordering of $C(E_t)$, and $\tilde{o_0}$ is the subsequence of $o_0$ containing only those operations not in $C(E_t)$.\footnote{A subsequence is obtained from a sequence by removing any number of elements from it without changing the order of the remaining elements. The subsequence does not have to appear continuously in the original sequence.}
Trivially, $o_0 \cdot o_1$ is a topological ordering of $C_t$.
Moreover, since there are no edges from nodes in $P(E_t)$ to nodes in $C(E_t)$, we can choose $c$, the topological ordering of $C(E_t)$, such that $c \cdot \tilde{o_0} \cdot o_1$ is also a topological ordering of $C_t$.
This allows us to decompose the topological ordering $o_0 \cdot o_1$ into two parts: $c$, a topological ordering of $C(E_t)$, and $\tilde{o_0} \cdot o_1$, a topological ordering of $P(E_t) \cap o$. Using $o_0 \simeq c \cdot \tilde{o_0}$, we conclude $o = o_0 \cdot op \cdot o_1 \simeq c \cdot \tilde{o_0} \cdot op \cdot o_1$. It remains to show that
\begin{align*}
    c \cdot \tilde{o_0} \cdot op \cdot o_1 \simeq c \cdot \tilde{o_0} \cdot o_1 \cdot op,
\end{align*}
which implies $o \simeq o_0 \cdot o_1 \cdot op$, i.e., that $op$ can be delayed to the end of the sequence without affecting the resulting state or any response value. Once this is established, the induction hypothesis for $C_{t+1}$ follows from the induction hypothesis for $C_t$ and the transitivity of the equivalence relation in the following way. Consider two prefixes $o = o_0 \cdot op \cdot o_1$ and $o' = o_0' \cdot op \cdot o_1'$ of topological orderings of $C_{t+1}$, and assume that both prefixes consist of the same set of operations. We conclude
\begin{align*}
    o = o_0 \cdot op \cdot o_1 \simeq o_0 \cdot o_1 \cdot op \simeq o_0' \cdot o_1' \cdot op \simeq o_0' \cdot op \cdot o_1' = o'.
\end{align*}
The equivalence $c \cdot \tilde{o_0} \cdot op \cdot o_1 \simeq c \cdot \tilde{o_0} \cdot o_1 \cdot op$ is shown by Lemma \ref{lemma:equivalence-of-any-placement} below.
\end{proof}

\begin{lemma}
    \label{lemma:equivalence-of-any-placement}
    Suppose that an operation $op$ is committed for the first time in the atomic extension $E_t := t \rightarrow t+1$, and the statement of Lemma~\ref{lemma:equivalence-of-every-topological-ordering} holds for every $t'$ with $0 \leq t' \leq t$.\footnote{This is the inductive hypothesis on $C_t$ in the proof of the lemma above.}

    Let $o$ be a prefix of a topological ordering of $C_{t+1}$ that contains $op$ and let $c$ be a topological ordering of $C(E_t)$ that is also a prefix of some topological ordering of $C_t$.

    For any fixed topological ordering $w$ of $P := P(E_t) \cap o$, all placements of $op$ in $w$ after $c$ are equivalent, i.e., $c \cdot w_1 \cdot op \cdot w_2 \simeq c \cdot w_1' \cdot op \cdot w_2'$ for all $w_1 \cdot w_2 = w_1' \cdot w_2' = w$.
\end{lemma}
\begin{proof}
We analyze the two cases, depending on whether $op$ has been committed via the conflict-free or the conflict resolution path, separately. 
We refer to the statement of Lemma~\ref{lemma:equivalence-of-every-topological-ordering} as the \textbf{Outer Induction Hypothesis}.

\textbf{Case~1: Conflict resolution path.}
Let $op$ be committed via the conflict resolution path. There exists an iteration $\hat{k}$ such that $\textit{CONS}_{\hat{k}}$ stores $op$. 
Let $p$ be the invocation of $publish$ that committed $op$.
As $p$ can reach iteration $\hat{k}$ only if all iterations $k < \hat{k}$ are completed and the outcomes of $\textit{CONS}_k.\textit{propose}()$ are added to $C$ (lines~\ref{line:uc:rc:read-k} and~\ref{line:uc:rc:write-k}), $C(E_t)$ read by $p$ in line~\ref{line:uc:rc:read-2} includes all operations stored in $\textit{CONS}_k$ with $k < \hat{k}$.
Analogously, an operation committed in a subsequent iteration $k > \hat{k}$ will include $op$ in its dependencies. 
Thus, as $P(E_t)$ is defined as the subgraph of $C_t$ induced by $\textit{vertices}(C_t)-\textit{vertices}(C(E_t))$ and $op\notin C_t$, every $op' \in P(E_t)$ has been committed via the conflict-free path.

Let $\hat{p}$ be an invocation of $publish$ that \emph{proposed} $op$ to $\textit{CONS}_{\hat{k}}$ such that no proposal on that consensus object precedes the proposal of $\hat{p}$. Note that $\hat{p}$ could be different from $p$.
In line~\ref{line:uc:rc:select}, $\hat{p}$ selected an operation, i.e., $op$, from $B_{\hat{t}} - \textit{vertices}(C_{\hat{t}})$ with minimal $B$-value according to the ABC Graph at time $\hat{t}$, the time of $\hat{p}$'s read in line~\ref{line:uc:rc:read-1}.

Consider an arbitrary invocation $p'=\textit{publish}(op')$ such that $op' \notin \textit{vertices}(C_{\hat{t}})$.
If $op' \notin B_{\hat{t}}$, then $p'$ will encounter $op \in A$ in line~\ref{line:uc:read-2}, because $op \in \textit{vertices}(B_{\hat{t}}) \subseteq A_{\hat{t}}$.
If $op' \in B_{\hat{t}}$, then $B[op'] \geq B[op]$, since $op$ has minimal $B$-value for time $\hat{t}$. Hence, $p'$ had or will encounter $op \in A$ in line~\ref{line:uc:read-2}.
We conclude, that every invocation $p'=\textit{publish}(op')$, such that $op' \notin \textit{vertices}(C_{\hat{t}})$, encounters $op \in A$ in line~\ref{line:uc:read-2}.

Since $\hat{t}$, the time when $\hat{p}$ read in line~\ref{line:uc:rc:read-1}, precedes $\hat{p}$'s proposal in line~\ref{line:uc:rc:consensus} and no proposal precedes $\hat{p}$'s proposal, the read of $C(E_t)$ in line~\ref{line:uc:rc:read-2} succeeds $\hat{t}$.
By Lemma~\ref{lemma:basic-properties}~(\ref{lemma:basic-properties.increasing}), $C_{\hat{t}} \subseteq C(E_t)$ follows.
Let $op' \in P(E_t)$ be an arbitrary operation.
As established above, $op'$ was committed via the conflict-free path before $op$ was committed.
Since $op' \in C_{\hat{t}}$ implies $op' \in C(E_t)$, which is a contradiction, we have $op' \notin C_{\hat{t}}$.
As we have shown above, $op'$ has found $op \in A$ in line~\ref{line:uc:read-2}.
Since $op'$ could not have found $op \in C$, we conclude that $op'$ has found $op \in A - \textit{vertices}(C)$ in line~\ref{line:uc:read-2}
and verified commutativity in line~\ref{line:uc:commutativity}.

Recall that $o$ is defined as a prefix of a topological ordering of $C_{t+1}$ that contains $op$, and that $c$ is defined as a topological ordering of $C(E_t)$ that is also a prefix of some topological ordering of $C_t$.
Let $w$ be a topological ordering of $P := P(E_t) \cap o$. 
In the following we are showing that all placements of $op$ in $w$ after $c$ are equivalent.
We prove this by induction over $P$, starting with an empty set and  
inductively adding  operations to $P$ one by one.
To define the induction order, we sort the operations according to the recency of the state that the corresponding invocation of $\textit{publish}$ read in line \ref{line:uc:read-2}, from the oldest to the most recent. 
Note that we insert operations in their original place in $w$ such that the relative ordering of $w$ remains the same. 
We introduce the following notation. 
Let $w^j$ denote the sequences of operations that consists of the first $j$ operations according to the sorting of $P$ that we just defined. 
Our induction hypothesis is the following:

\textbf{Inner Induction Hypothesis:} For every fixed prefix $\hat{w}^j$ of $w^j$, all placements of $op$ in $\hat{w}^j$ after $c$ are equivalent.

\textbf{Inner Induction Base Case:} The base case is trivial as $w^0$, and therefore every prefix $\hat{w}^0$, is empty.

\textbf{Inner Induction Step:} Assume that the induction hypothesis holds true for $0 \leq j < |P|$. 
Let $\hat{w}^{j+1}$ be a prefix of $w^{j+1}$. 
If $\hat{w}^{j+1}$ is also a prefix of $w^j$, the inner induction hypothesis for $j$ applies and we are done. 
Let $\hat{w}^{j+1} = u^j \cdot op' \cdot v^j$ such that $u^j \cdot v^j$ is a prefix of $w^j$ and $op' \in P$ is the next operation, added in this inner induction step. 
We have to show, that all placements of $op$ in $\hat{w}^{j+1}$ after $c$ are equivalent.
This is done by showing that every placement of $op$ is equivalent to placing it after $\hat{w}^{j+1}$. 
The induction hypothesis for $j+1$ follows by transitivity.

By definition of $\hat{w}^{j+1}$, we have $c \cdot \hat{w}^{j+1} \cdot op = c \cdot u^j \cdot op' \cdot v^j \cdot op$.
As $c \cdot u^j$ is a prefix of some topological ordering of the graph $C_t$ we can use the Outer Induction Hypothesis.\footnote{By definition, $c$ is a topological ordering of $C(E_t)$ and a prefix of some topological ordering of $C_t$. It is easy to see that $c \cdot w$ is also a prefix of a topological ordering of $C_t$.
When removing operations from $w$ in the reverse order in which they were introduced during the induction, we always remove the operation that read the most recent state in line~\ref{line:uc:read-2}, which therefore cannot have an outgoing edge to any of the remaining operations.
By induction, $c \cdot w^j$ is a prefix of a topological ordering of $C_t$; since $u^j$ is a prefix of $\hat{w}^j$, which in turn is a prefix of $w^j$, the same holds for $c \cdot \hat{w}^j$ and $c \cdot u^j$.}
Let $c_{op'} \cdot u^j_{op'}$ be a topological ordering of the same prefix of $C_t$, where $c_{op'}$ are the operations that $\textit{publish}(op')$ read as $C$ in line \ref{line:uc:read-2}. 
The Outer Induction Hypothesis gives us $c \cdot u^j \simeq c_{op'} \cdot u^j_{op'}$. 
Note that $u^j_{op'} \cdot op' \cdot v^j \cdot op$ is a subset of $\textit{CUR}$ (the operations in $A-\textit{vertices}(C)$) of $\textit{publish}(op')$ in line \ref{line:uc:read-2}. 
Lemma \ref{lemma:equivalence-preservation} implies that all placements of $op'$ in this set are equivalent. It follows that
\begin{align*}
    c \cdot \hat{w}^{j+1} \cdot op
    &= c \cdot u^j \cdot op' \cdot v^j \cdot op \\
    &\simeq c_{op'} \cdot u^j_{op'} \cdot op' \cdot v^j \cdot op \\
    &\simeq c_{op'} \cdot u^j_{op'} \cdot v^j \cdot op \cdot op' \\
    &\simeq c \cdot u^j \cdot v^j \cdot op \cdot op'.
\end{align*}
We distinguish two cases, according to whether $op$ should be placed in $u^j$ or $v^j$.

\textbf{Case~1.1:} To place $op$ in $u^j = (u^j)_0 \cdot (u^j)_1$, where $(u^j)_0 \cdot (u^j)_1$ is an arbitrary split of $u^j$, we use the Inner Induction Hypothesis for $j$ in the following way. Note that $u^j \cdot v^j$ and therefore also $u^j$ are prefixes of $w^j$.
\begin{align*}
    c \cdot u^j \cdot v^j \cdot op \cdot op'
    &\simeq c \cdot u^j \cdot op \cdot v^j \cdot op' \\
    &\simeq c_{op'} \cdot u^j_{op'} \cdot op \cdot v^j \cdot op' \\
    &\simeq c_{op'} \cdot u^j_{op'} \cdot op \cdot op' \cdot v^j \\
    &\simeq c \cdot u^j \cdot op \cdot op' \cdot v^j \\
    &\simeq c \cdot (u^j)_0 \cdot op \cdot (u^j)_1 \cdot op' \cdot v^j.
\end{align*}

\textbf{Case~1.2:} Suppose that $op$ has to be placed in $v^j = (v^j)_0 \cdot (v^j)_1$, where $(v^j)_0 \cdot (v^j)_1$ is an arbitrary split of $v^j$.
\begin{align*}
    c \cdot u^j \cdot v^j \cdot op \cdot op'
    &\simeq c \cdot u^j \cdot (v^j)_0 \cdot op \cdot (v^j)_1 \cdot op' \\
    &\simeq c_{op'} \cdot u^j_{op'} \cdot (v^j)_0 \cdot op \cdot (v^j)_1 \cdot op' \\
    &\simeq c_{op'} \cdot u^j_{op'} \cdot op' \cdot (v^j)_0 \cdot op \cdot (v^j)_1 \\
    &\simeq c \cdot u^j \cdot op' \cdot (v^j)_0 \cdot op \cdot (v^j)_1.
\end{align*}

We conclude that every placement of $op$ in $w^{j+1}$ after $c$ is equivalent to placing $op$ after $c \cdot w^{j+1}$.
Transitivity of the equivalence relation gives us the Inner Induction Hypothesis for $j+1$.
The claim follows from the Inner Induction Hypothesis for $j = |P|$.

\textbf{Case~2: Conflict-free path.}
Assume now that $op$ has been committed via the conflict-free path. 
Let $w$ be a topological ordering of $P$. 
Recall that we have to show that every placement of $op$ in $w$ after $c$ is equivalent. 
As $op$ has been committed via the conflict-free path, commutativity with every subset of $\textit{CUR} \cap P$ has been checked. 
By Lemma~\ref{lemma:equivalence-preservation}, for any fixed ordering $p$ of $\textit{CUR} \cap P$, all placements of $op$ in $p$ after $c$ are equivalent, i.e., $c \cdot p_0 \cdot op \cdot p_1 \simeq c \cdot p_0' \cdot op \cdot p_1'$ for $p = p_0 \cdot p_1 = p_0' \cdot p_1'$.
We are therefore concerned about the operations that have been concurrently added and have not been encountered by $op$, i.e., $P - CUR$. 
Consider $op' \in P - CUR$.
By the definition of $\textit{CUR}$, $op'$ has not been read by $op$ in $A$ in line \ref{line:uc:read-2}.
Therefore, $op'$ could not have been proposed to any consensus object in conflict resolution, as $op'$ has a higher $B$-value than $op$, and $op$ has been booked prior to $op'$. 
From this we infer that every $op' \in P - CUR$ has been added via the conflict-free path. 
As $op'$ has not been found by $op$ in $A$ in line~\ref{line:uc:read-2}, $op'$ has read in line~\ref{line:uc:read-1} a (proper) superset of what $op$ has read in line~\ref{line:uc:read-2} (Lemma~\ref{lemma:basic-properties}~(\ref{lemma:basic-properties.increasing})).
In particular, $op'$ has found $op \in A - \textit{vertices}(C)$ in line~\ref{line:uc:read-2} and verified commutativity in line~\ref{line:uc:commutativity}.

To prove that all placements of $op$ in $w$ after $c$ are equivalent, we follow an argument analogous to that employed in the conflict resolution path case above.
In contrast to the conflict resolution path case, we do not need to begin with an empty set before incrementally adding operations. 
As shown above, we already got the equivalence of all placements of $op$ in the orderings of subsets of $\textit{CUR} \cap P$.
The operations in $P - CUR$ are handled as in the conflict resolution path case. 
We sort them according to the recency of the state they read in line \ref{line:uc:read-2}, from the oldest to the most recent. 
These operations are then inserted at their original positions in $w$, thereby preserving the relative ordering within $w$. 
We introduce the following notation. 
Let $w^j$ denote the sequences of operations that consists of the operations in $\textit{CUR} \cap P$ as well as the first $j$ operations according to the sorting of $P - CUR$ that we just defined. 
The Inner Induction Hypothesis and step is identical to the conflict resolution path case. 
The base case is directly given by Lemma~\ref{lemma:equivalence-preservation} as $op$ commutes with every subset of $\textit{CUR} \cap P$.
The claim follows directly from the Inner Induction Hypothesis for $j = |P-CUR|$.
\end{proof}

\subsection{Proof Completion}
\label{sec:appendix:uc:linearizability:proof-completion}
Lemmata~\ref{lemma:acyclicity}, \ref{lemma:commitment} and~\ref{lemma:precedence} correspond exactly to the first three items of Lemma~\ref{lemma:correctness}.
Having established Lemma~\ref{lemma:equivalence-of-every-topological-ordering} in the last section, we proceed to show the remaining two items of Lemma~\ref{lemma:correctness} in Lemmata~\ref{lemma:consistency} and~\ref{lemma:response-value-sanity}, thereby completing the proof of Lemma~\ref{lemma:correctness}.
Finally, we show in Lemma~\ref{lemma:linearizability} that Algorithm~\ref{algo:universal-construction} is linearizable.

\begin{lemma}
    \label{lemma:consistency}
    Let $t_1, t_2 \geq 0$, and let $s_1$ and $s_2$ be topological orderings of $C_{t_1}$ and $C_{t_2}$, respectively. For every operation $op$ such that $op \in s_1$ and $op \in s_2$, the sequential specification yields the same return value for $op$ in both $s_1$ and $s_2$.
\end{lemma}
\begin{proof}
    Assume, without loss of generality, $0 \leq t_1 \leq t_2$.
    Let $s_1$ and $s_2$ be topological orderings of $C_{t_1}$ and $C_{t_2}$, respectively.
    By Lemma~\ref{lemma:separation-of-orderings}, there exists some topological ordering $r = r_1 \cdot r_2$ of $C_{t_2}$ such that $r_1$ is a topological ordering of $C_{t_1}$.
    Lemma~\ref{lemma:equivalence-of-every-topological-ordering} then implies $r \simeq s_2$ and $r_1 \simeq s_1$.
    Since $op \in s_1$, we also have $op \in r_1$.
    With $r_1 \simeq s_1$, the response value for $op$ is identical in both $s_1$ and $r_1$, and therefore also in $r$.
    The claim follows because $r \simeq s_2$.
\end{proof}

\begin{lemma}
    \label{lemma:response-value-sanity}
    For every time $t_2 \geq 0$ and every operation $op \in C_{t_2}$, if $op$ has a response, then there exists a time $t_1$ with $0 \leq t_1 \leq t_2$ and a topological ordering $T$ of $C_{t_1}$ that contains $op$, such that the response of $op$ agrees with the sequential specification applied to $T$.
\end{lemma}
\begin{proof}
    Since $op$ has a response, the corresponding $\textit{publish}(op)$ terminated. If $\textit{publish}(op)$ returned at lines~\ref{line:uc:return-1} or~\ref{line:uc:rc:return}, the claim is immediate.
    Suppose that $\textit{publish}(op)$ returned at line~\ref{line:uc:return-2}.
    Let $t_{com} \rightarrow t_{com}+1$ be the atomic extension produced by $\textit{add}(op, M)$ in line~\ref{line:uc:addC} with $M=\textit{vertices}(C_{t_{read}})$ for some $t_{read} \leq t_{com}$ (Lemma~\ref{lemma:previous-nodes}).
    Let $l$ be the topological ordering of $C_{t_{read}}$ that $\textit{publish}(op)$ calculated in line~\ref{line:uc:linearize-C}.

    \textbf{Case~1:} Suppose $op \notin C_{t_{com}}$.
    By Lemma~\ref{lemma:separation-of-orderings}, there exists a topological ordering $r = r_1 \cdot r_2$  of $C_{t_{com}}$ such that $r_1$ is a topological ordering of $C_{t_{read}}$, and $r_2$ consists of nodes from $\textit{vertices}(C_{t_{com}}) - \textit{vertices}(C_{t_{read}})$.
    Since $op \notin C_{t_{com}}$, there is no edge in $C_{t_{com}+1}$ from a node in $\textit{vertices}(C_{t_{com}+1}) - \textit{vertices}(C_{t_{read}})$ to $op$.
    Hence, $T := r_1 \cdot op \cdot r_2$ is a topological ordering of $C_{t_{com}+1}$.
    By Lemma~\ref{lemma:equivalence-of-every-topological-ordering}, $r_1$ is equivalent to $l$, and thus $T$ is equivalent to $l \cdot op \cdot r_2$.
    Consequently, the response of $op$ (line~\ref{line:uc:return-2}) agrees with the sequential specification applied to $T$.

    \textbf{Case~2:} Suppose that $op \in C_{t_{com}}$.
    Therefore, there was a preceding commitment of $op$.
    Let $p$ be the process $\textit{publish}(op).p$ and $p'$ the process that committed $op$ for the first time.
    For the sake of the argument, in the following, we will introduce alternatives to the original run, that we denote by $\textit{ORUN}$.
    For a run $R'$ and a process $q$ that can do a step, $R'.q$ denotes the run where $q$ makes a step after $R'$, and no other step is made.
    Since the ABC Graph is linearizable, we view each operation on it as one single step (together with local computations).
    Furthermore, let $C_{R'}$ denote the graph $C$ after the run $R'$.

    Let $R$ be the prefix of $\textit{ORUN}$ right before $p'$ commits $op$.
    Since $p$ has not found $op$ after executing line~\ref{line:uc:read-2}, it is also about to commit $op$ in the next step and return in line~\ref{line:uc:return-2}.
    Consequently, $R.p$, $R.p'.p$ and $\textit{ORUN}$ are indistinguishable for $p$ at the time when $p$ terminates $\textit{publish}(op)$.
    In particular, the response value of $op$ is the same in all runs.
    An argument analogous to the one used in Case~1 shows that there is a topological ordering $T' := r_1 \cdot op \cdot r_2$ of $C_{R.p}$ with $r_1 \simeq l$.
    Therefore, the response of $op$ (line~\ref{line:uc:return-2}) agrees with the sequential specification applied to $T'$.

    Since $R.p$ is a prefix of $R.p.p'$, the graphs $C_{R.p}$ and $C_{R.p.p'}$ belong to the same run $R.p.p'$\footnote{Let $\hat{t} \geq 2$ be the time at the end of run $R.p.p'$. The graph $C_{R.p}$ appears in run $R.p.p'$ at time $\hat{t}-1$.} and therefore Lemma~\ref{lemma:consistency} is applicable: The sequential specification yields the same response value for $op$ in every topological ordering of $C_{R.p}$ and $C_{R.p.p'}$.
    The same holds for $C_{R.p'}$ and $C_{R.p'.p}$ analogously. 
    Furthermore, we observe that, as every two $\textit{add}$ operations on the ABC Graph commute, the graphs $C_{R.p.p'}$ and $C_{R.p'.p}$ are identical.
    We conclude that the sequential specification yields the same response for $op$ in every topological ordering of $C_{R.p}$ and $C_{R.p'}$.

    Since the response of $op$ agrees with the sequential specification applied to a topological ordering of $C_{R.p}$, namely $T'$, it also agrees with every topological ordering of $C_{R.p'}$.
    Since $R.p'$ is effectively a prefix of the original run $\textit{ORUN}$, $C_{R.p'}$ is a version of graph $C$ in $\textit{ORUN}$. The claim follows directly.
\end{proof}

\begin{lemma}
    \label{lemma:linearizability}
    Algorithm~\ref{algo:universal-construction} is a linearizable universal construction. Moreover, for any run of the algorithm, every topological ordering of the final graph $C$ is a linearization of the induced history.
\end{lemma}
\begin{proof}
    To prove the claim, we have to show that, for any given sequential specification and any finite run of Algorithm~\ref{algo:universal-construction}, a topological ordering of the final graph $C$ is a linearization of the induced history.\footnote{Note that the implementation of universal construction presented in Algorithm~\ref{algo:universal-construction} exports the method $\textit{publish}(op)$, whereas the object specified by the sequential specification exports the method $op$ directly. This difference does not affect correctness and can be disregarded when reasoning about equivalences of histories.}
    Let $H$ be the history of a run $R$ of Algorithm~\ref{algo:universal-construction}.
    Let $C_R$ denote the graph $C$ at the end of run $R$. Let $S$ be a topological ordering of $C_R$. Lemma~\ref{lemma:acyclicity} guarantees the existence of such a topological ordering.
    As a sequence of operations, $S$ forms a sequential history.
    It remains to show that $S$ serves as a linearization of $H$.
    By Definition~\ref{def:history.linearization}, $S$ must be equivalent\footnote{Note the distinction between equivalence of histories (Definition~\ref{def:history.equivalent}) and equivalence of sequences of operations (Definition~\ref{def:equivalence-relation-of-orderings}).} to a completion of $H$, preserve the precedence relation of $H$ and be legal.

    By Lemmata~\ref{lemma:commitment} and \ref{lemma:precedence}, $S$ is equivalent to a completion of $H$ and preserves the precedence relation of $H$.
    It remains to show legality, i.e., that $S$ satisfies the sequential specification.
    For each operation $op$ in $S$, the response must agree with the sequential specification applied to $S$. For operations in $S$ that do not have a response, we can assign one according to the sequential specification applied to $S$, which trivially satisfies legality.

    Let $op$ be an operation in $S$ that has a response.
    By Lemma~\ref{lemma:response-value-sanity}, there exists a topological ordering $T$ of $C_t$ for some time $t \geq 0$ that contains $op$, such that the response of $op$ agrees with the sequential specification applied to $T$.
    By Lemma~\ref{lemma:consistency}, the sequential specification yields the same response for $op$ in both $T$ and $S$, i.e., the response of $op$ agrees with the sequential specification applied to $S$.
    
    Since this holds for every operation in $S$, each operation returned the value prescribed by the sequential specification in $S$. Hence, $S$ is legal. We conclude that $S$ is a linearization of $H$, and therefore that Algorithm~\ref{algo:universal-construction} is linearizable.
\end{proof}
\section{Universal Construction - Proof of Wait-Freedom}
\label{sec:appendix:uc:wait-freedom}
Having established the linearizability of Algorithm~\ref{algo:universal-construction} in the previous section, we now address its liveness guarantees, specifically proving in Lemma~\ref{lemma:wait-freedom} that it satisfies wait-freedom, assuming that the underlying shared objects itself are wait-free. The underlying objects are a snapshot object $K$ (line~\ref{line:uc:snapshot-object}), an ABC Graph $G$ (line~\ref{line:uc:abc-graph-object}), and a sequence of consensus objects $\textit{CONS}_j$ for $j \geq 1$ (line~\ref{line:uc:consensus-objects}).
We showed in Section~\ref{sec:appendix:abc-graph}, that Algorithm~\ref{algo:abc-graph} is a wait-free linearizable implementation of an ABC Graph.

\begin{lemma}
    \label{lemma:wait-freedom}
    The universal construction given in Algorithm~\ref{algo:universal-construction} is wait-free. Each invocation of $\textit{publish}$ terminates at most after $c + 2$ iterations of conflict resolution, for the maximum level of concurrency $c \leq n$.
\end{lemma}
\begin{proof}
Let $p$ be a process that executes $\textit{publish}(op)$.

\textbf{Case~1:} Suppose $p$ verifies commitment of $op$ in line~\ref{line:uc:checkC} or the commutativity in line~\ref{line:uc:commutativity}. It is immediate that the number of steps it can take before terminating $\textit{publish}(op)$ is finite.

\textbf{Case~2:} Suppose $p$ fails to verify the commitment of $op$ in line~\ref{line:uc:checkC} and fails to verify the commutativity property in line~\ref{line:uc:commutativity}. Let $k_0$ denote the value of $k \geq 0$ read by $p$ in line~\ref{line:uc:rc:read-k} and let $c \leq n$ be the maximum level of concurrency, i.e., the maximum number of processes that can propose operations concurrently. We prove by contradiction that $\textit{publish}(op)$ terminates at latest in iteration $k_0 + c + 2$.

Assume, for the sake of contradiction, that $\textit{publish}(op)$ does not terminate by iteration $k_0 + c + 2$. Since it did not return in iteration $k_0 + 1$, $op$ was not decided by any $\textit{CONS}_k$ with $k \leq k_0$; otherwise, $op$ would have been committed and $\textit{publish}(op)$ would have returned at line~\ref{line:uc:rc:return} in iteration $k_0 + 1$. From iteration $k_0 + 1$ onward, $\textit{publish}(op)$ executes the conflict resolution procedure; hence, each $k$ with $k_0 + 1 \leq k \leq k_0 + c + 1$, $\textit{CONS}_k$ decides on an operation different from $op$. Otherwise, $\textit{publish}(op)$ would terminate in iteration $k + 1$.

Let us consider the $c$ different iterations with $k_0 + 2 \leq k \leq k_0 + c + 1$ and the operations stored in $\textit{CONS}_k$. Since $c$ is the concurrency bound, among those $c$ operations there exists an operation $op'$ such that $\textit{publish}(op')$ was not active at the moment when $p$ read $k_0$. Either $\textit{publish}(op')$ terminated before $p$ read $k_0$ or $\textit{publish}(op')$ was invoked afterwards.

\textbf{Case~2.1:} Suppose $\textit{publish}(op')$ was invoked after $p$ read $k_0$. The $B$-value of $op'$ will be higher than $B[op]$. Furthermore, every process that encounters $op'$ in line~\ref{line:uc:rc:read-1} will also encounter $op$ with a lower $B$-value. Consequently, no process would propose $op'$ to a consensus object $\textit{CONS}_k$ with $k_0 + 2 \leq k \leq k_0 + c + 1$. This contradicts the validity of $\textit{CONS}_k$.

\textbf{Case~2.2:} Suppose $\textit{publish}(op')$ terminated before $p$ read $k_0$. By Lemma~\ref{lemma:commitment}, $op'$ is committed before $p$ read $k_0$ and consequently before any process executed line~\ref{line:uc:rc:read-1} in iteration $k_0 + 2$; otherwise, $p$ would have read at least $k_0 + 1$ in line~\ref{line:uc:rc:read-k}. Consequently, no process would propose $op' \in C$ to a subsequent consensus object, contradicting validity of consensus.

We conclude that the assumption that $\textit{publish}(op)$ does not terminate by iteration $k_0 + c + 2$ leads to a contradiction. Therefore, Algorithm~\ref{algo:universal-construction} is wait-free.
\end{proof}

\section{Universal Construction - Proof of Dynamic Concurrency}
\label{sec:appendix:uc:dynamic-concurrency}
After proving linearizability and wait-freedom of Algorithm~\ref{algo:universal-construction}, we show in Lemma~\ref{lemma:dynamic-concurrency} that the algorithm satisfies dynamic concurrency according to Definition~\ref{def:dynamic-concurrency}.

\begin{lemma}
    \label{lemma:dynamic-concurrency}
    The universal construction given in Algorithm~\ref{algo:universal-construction} is dynamically concurrent.
\end{lemma}
\begin{proof}
    Let $H$ be an arbitrary history induced by a run of Algorithm~\ref{algo:universal-construction} and $\textit{publish}(op)$ an operation in $H$. 
    Assume that process $\textit{publish}(op).p$ used a strong synchronization primitive while executing $\textit{publish}(op)$ and, thus, executed conflict resolution for $op$.
    Therefore, the conditions in lines~\ref{line:uc:checkC} and~\ref{line:uc:commutativity} failed for $\textit{publish}(op)$.

    Let $t'$ be the time when $\textit{publish}(op).p$ read the ABC Graph in line~\ref{line:uc:read-2}. It remains to show that there exists a set of operations $O$ with $O'' := A_{t'} - \textit{vertices}(C_{t'}) - \{op\} \subseteq O$ such that $(l, O)$ is a configuration during $op$ in $H$ for $l := \textit{linearize}(C_{t'})$.
    %[[PK you mentioned this before
    %\footnote{As seen before, the universal construction in Algorithm~\ref{algo:universal-construction} exports the method $\textit{publish}(op)$ instead of $op$ directly. This formality is omitted in the proof of Lemma~\ref{lemma:dynamic-concurrency} to enhance readability.}
    %]]

    Let $C_H$ be the graph $C$ at the end of the run that induced $H$. 
    By Lemma~\ref{lemma:linearizability}, every topological ordering of $C_H$ is a linearization of $H$. 
    By Lemma~\ref{lemma:no-backwards-edge}, no operation in $\textit{vertices}(C_H)-\textit{vertices}(C_{t'})$ has an edge to an operation in $\textit{vertices}(C_{t'})$. 
    Therefore, there exists a topological ordering $\hat{l} = \hat{l}_0 \cdot \hat{l}_1$ of $C_H$ and thus a linearization of $H$ such that $op' \in \hat{l}_0$ if and only if $op' \in C_{t'}$ for all $op' \in C_H$.
    Since $\hat{l}_0$ is also a topological ordering of $C_{t'}$, we have $l \simeq \hat{l}_0$ by Lemma~\ref{lemma:equivalence-of-every-topological-ordering}. 
    Thus, $L := l \cdot \hat{l}_1$ is a linearization of $H$ and $l$ a prefix of $L$.
    Furthermore, since $\textit{publish}(op).p$ executed conflict resolution, it failed the condition $op \in C_{t'}$ in line~\ref{line:uc:checkC}. 
    We therefore have $op \notin l$.

    Let $H'$ be the maximal prefix of $H$ such that $op'.res \in H'$ implies $op' \in l$ for every operation $op'$. 
    We have to show $op.inv \in H'$ and $op.res \notin H'$.
    Since $H'$ is maximal, it excludes $op.inv$ only if there exists an operation $op' \notin l$ with $op'.res <_H op.inv$. 
    Let $op'$ be an operation such that $op'.res$ precedes $op.inv$ in $H$. 
    By Lemma~\ref{lemma:commitment}, $op'$ is committed upon termination of $\textit{publish}(op')$. 
    Therefore, we have $op' \in C_{t'}$ and consequently $op' \in l$. Thus, $op.inv \in H'$ holds. 
    Finally, $op.res \notin H'$ is implied by $op \notin l$.
    
    Furthermore, we have to show that for all $op' \in l$, it holds that $op'.inv \in H'$. 
    Suppose, for the sake of contradiction, that there exists $op' \in l$ with $op'.inv \notin H'$. 
    Since $H'$ is maximal, $op'.inv \notin H'$ implies that there exists $op'' \notin l$ with $op''.res <_H op'.inv$. 
    By Lemma~\ref{lemma:precedence}, $op''$ precedes $op'$ in every topological ordering of $C_{t'}$. This contradicts $op'' \notin l$. 
    We conclude that $op'.inv \in H'$ for all $op' \in l$.

    We now have to show that $l$ is a linearization of $H'$. 
    By construction, $l$ is a prefix of $L$, where $L$ is a linearization of $H$. 
    Consequently, $l$ is a legal sequential history. 
    Since $H'$ is a prefix of $H$, $l$ preserves the precedence relation of $H'$, as the precedence relation of $H'$ is a subset of that of $H$. 
    It remains to show that $l$ is equivalent to a completion of $H'$.
    As established above, we have $op'.inv \in H'$ for all $op' \in l$ and, by construction, $op' \in l$ for all $op.res \in H'$. 
    Hence, we can construct a completion $\hat{H'}$ of $H'$ such that $op' \in \hat{H'}$ if and only if $op' \in l$. 
    It remains to show that $l$ is equivalent to $\hat{H'}$ according to Definition~\ref{def:history.equivalent}.
    Suppose, for the sake of contradiction, that $l$ and $\hat{H'}$ are not equivalent. 
    Hence, there exist some process $p'$ and operations $op', op'' \in \hat{H'} \mid p'$ such that $op' <_{\hat{H'}} op''$ but $op'' <_l op'$. 
    This leads to a contradiction, since $op' <_{\hat{H'}} op''$ implies $op' <_{H'} op''$, which in turn implies $op' <_H op''$, then $op' <_L op''$, and finally $op' <_l op''$. 
    Therefore, $l$ is equivalent to $\hat{H'}$, a completion of $H'$, and consequently, $l$ is a linearization of $H'$.

    It remains to show that $O'' = A_{t'} - \textit{vertices}(C_{t'}) - \{op\}$ is a subset of the set $O$ of all operations $op'$ for which $op'.inv \in H'$, $op' \notin l$, and $op' \neq op$. 
    Let $op' \in O''$. 
    By definition of $O''$, we have $op' \neq op$ and $op' \notin \textit{vertices}(C_{t'})$ and therefore $op' \notin l$. 
    Suppose, for the sake of contradiction, that $op'.inv \notin H'$. 
    By definition of $H'$, there must exist an operation $op'' \notin l$ such that $op''.res$ precedes $op'.inv$; otherwise, $H'$ would not be maximal. 
    It follows from $op' \in O'' \subseteq A_{t'}$ that time $t'$, when $\textit{publish}(op).p$ read the ABC Graph in line~\ref{line:uc:read-2}, succeeds $op'.inv$ and therefore $op''.res$. 
    Lemma~\ref{lemma:commitment} contradicts $op'' \notin C_{t'}$, i.e., $op'' \notin l$. 
    This contradiction proves $O'' \subseteq O$. 
    The claim follows from Definition~\ref{def:dynamic-concurrency}.
\end{proof}

\end{document}